\documentclass[10pt,journal,compsoc]{IEEEtran}
\ifCLASSOPTIONcompsoc
  \usepackage[nocompress]{cite}
\else
  \usepackage{cite}
\fi
\ifCLASSINFOpdf
\else
\fi
\usepackage[numbers,sort&compress]{natbib}
\usepackage{graphicx}
\usepackage{amsmath,amsfonts,amssymb,amsthm} 
\usepackage{algorithm}
\usepackage{algorithmic}
\usepackage{subfigure}
\renewcommand{\Re}{\mathbb{R}}
\newcommand{\Id}{\mathtt{I}}

\newcommand{\trace}[1]{\ensuremath{\mathrm{tr}\left(#1\right)}}
\newcommand{\diag}[1]{\ensuremath{\mathrm{diag}\left(#1\right)}}

\newcommand{\rank}[1]{\ensuremath{\mathrm{rank}\left(#1\right)}}
\newcommand{\norm}[1]{\ensuremath{\left\| #1 \right\|}}

\newtheorem{prob}{Problem}[section]
\newtheorem{theo}{Theorem}[section]
\newtheorem{defn}{Definition}[section]
\newtheorem{lemm}{Lemma}[section]

\begin{document}
\title{Matrix Completion with Deterministic Sampling: Theories and Methods}
\author{Guangcan Liu,~\IEEEmembership{Senior Member,~IEEE}, Qingshan Liu,~\IEEEmembership{Senior Member,~IEEE},\\ Xiao-Tong Yuan,~\IEEEmembership{Member,~IEEE}, and Meng Wang
\IEEEcompsocitemizethanks{\IEEEcompsocthanksitem G. Liu is with B-DAT and CICAEET, School of Automation, Nanjing University of
Information Science and Technology, Nanjing, China 210044. Email: gcliu@nuist.edu.cn.\protect
\IEEEcompsocthanksitem Q. Liu is with B-DAT, School of Automation, Nanjing University of
Information Science and Technology, Nanjing, China 210044. Email: qsliu@nuist.edu.cn.\protect
\IEEEcompsocthanksitem X-T. Yuan is with B-DAT and CICAEET, School of Automation, Nanjing University of
Information Science and Technology, Nanjing, China 210044. Email: xtyuan@nuist.edu.cn.\protect
\IEEEcompsocthanksitem M. Wang is with School of Computer Science and Information Engineering, Hefei University of Technology, 193 Tunxi Road, Hefei, Anhui, China 230009. E-mail: eric.mengwang@gmail.com.\protect
}
\thanks{}}

\markboth{IEEE Transactions on Pattern Analysis and Machine Intelligence,~Vol.~XX,~No.~XX, 2019}%
{Shell \MakeLowercase{\textit{et al.}}: Bare Demo of IEEEtran.cls for Computer Society Journals}

\maketitle

\begin{abstract}
In some significant applications such as data forecasting, the locations of missing entries cannot obey any non-degenerate distributions, questioning the validity of the prevalent assumption that the missing data is randomly chosen according to some probabilistic model. To break through the limits of random sampling, we explore in this paper the problem of real-valued matrix completion under the setup of deterministic sampling. We propose two conditions, isomeric condition and relative well-conditionedness, for guaranteeing an \emph{arbitrary} matrix to be recoverable from a sampling of the matrix entries. It is provable that the proposed conditions are weaker than the assumption of uniform sampling and, most importantly, it is also provable that the isomeric condition is \emph{necessary} for the completions of \emph{any} partial matrices to be identifiable. Equipped with these new tools, we prove a collection of theorems for missing data recovery as well as convex/nonconvex matrix completion. Among other things, we study in detail a Schatten quasi-norm induced method termed isomeric dictionary pursuit (IsoDP), and we show that IsoDP exhibits some distinct behaviors absent in the traditional bilinear programs.
\end{abstract}

\IEEEdisplaynotcompsoctitleabstractindextext
\IEEEpeerreviewmaketitle

\section{Introduction}
\IEEEPARstart{I}{n} the presence of missing data, the representativeness of data samples may be reduced significantly and the inference about data is therefore distorted seriously. Given this pressing circumstance, it is crucially important to devise computational methods that can restore unseen data from available observations. As the data in practice is often organized in matrix form, it is considerably significant to study the problem of \emph{matrix completion}~\cite{tao:2009:mc,CandesPIEEE,Mohan:2010:isit,rahul:jlmr:2010,akshay:2013:nips,william:2014:nips,raghunandan:jmlr:2010,raghunandan:tit:2010,troy:2013:nips}, which aims to fill in the missing entries of a partially observed matrix.
\begin{prob}[Matrix Completion]\label{pb:mc}
Denote by $[\cdot]_{ij}$ the $(i,j)$th entry of a matrix. Let $L_0\in\Re^{m\times{}n}$ be an unknown matrix of interest. The rank of $L_0$ is unknown either. Given a sampling of the entries in $L_0$ and a 2D sampling set $\Omega\subseteq{}\{1,\cdots,m\} \times\{1,$ $\cdots,n\}$ consisting of the locations of observed entries, i.e., given
\begin{eqnarray*}
\Omega\quad\textrm{and}\quad\{[L_0]_{ij} |(i,j)\in\Omega\},
\end{eqnarray*}
can we identify the target $L_0$? If so, under which conditions?
\end{prob}
\begin{figure}
\begin{center}
\includegraphics[width=0.48\textwidth]{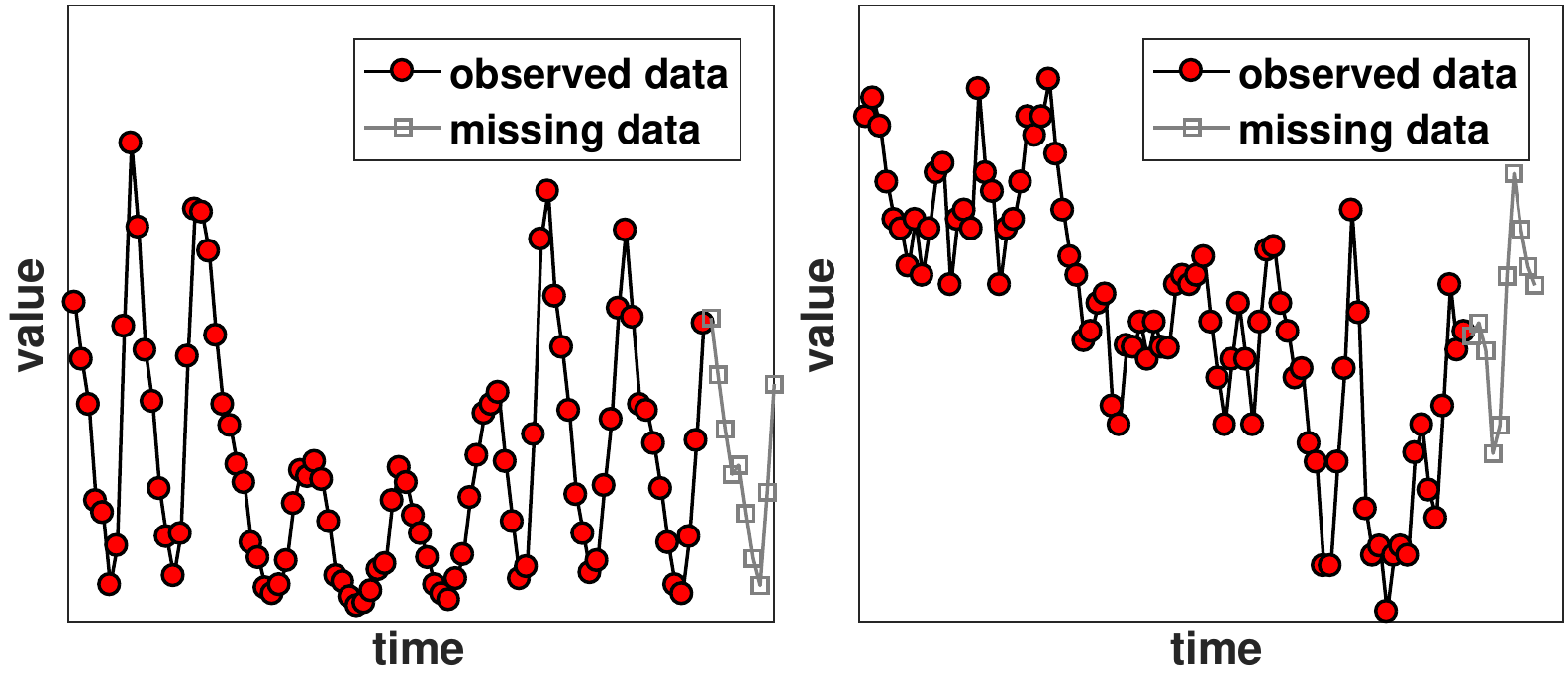}\vspace{-0.15in}
\caption{The unseen future values of time series are essentially a special type of missing data.}\label{fig:miss}\vspace{-0.25in}
\end{center}
\end{figure}

In general cases, matrix completion is an ill-posed problem, as the missing entries can be of arbitrary values. Thus, some assumptions are necessary for studying Problem~\ref{pb:mc}. Cand{\`e}s and Recht~\cite{Candes:2009:math} proved that the target $L_0$, with high probability, is exactly restored by convex optimization, provided that $L_0$ is \emph{low rank} and \emph{incoherent} and the set $\Omega$ of locations corresponding to the observed entries is a set sampled \emph{uniformly at random} (i.e., uniform sampling). This pioneering work provides people several useful tools to investigate matrix completion and many other related problems. Its assumptions, including low-rankness, incoherence and uniform sampling, are now standard and widely used in the literatures, e.g.,~\cite{Candes:2009:JournalACM,xu:2012:tit,sun:2016:tit,tpami_2013_lrr,Jain:2014:nips,liu:tpami:2016,zhao:nips:2015,ge:nips:2016}. However, the assumption of uniform sampling is often invalid in practice:
\begin{itemize}
\item[$\bullet$] A ubiquitous type of missing data is the unseen future data, e.g., the next few values of a time series as shown in Figure~\ref{fig:miss}. It is certain that the (missing) future data is not randomly selected, not even being sampled uniformly at random. In this case, as will be shown in Section~\ref{sec:exp:rcn}, the theories built upon uniform sampling are no longer applicable.
\item[$\bullet$] Even when the underlying regime of the missing data pattern is a probabilistic model, the reasons for different observations being missing could be correlated rather than independent. In fact, most real-world datasets cannot satisfy the uniform sampling assumption, as pointed out by~\cite{ruslan:2010:nips,Meka:2009:MCP}.
\end{itemize}

There has been sparse research in the direction of deterministic or nonuniform sampling, e.g.,~\cite{Kiraly:2012:icml,Kiraly:2015:jmlr,Negahban:2012:JMLR,ruslan:2010:nips,Meka:2009:MCP,JMLR:v16:chen15b,daniel:2016:jstsp}. For example, Negahban and Wainwright~\cite{Negahban:2012:JMLR} studied the case of weighted entrywise sampling, which is more general than the setup of uniform sampling but still a special form of random sampling. In particular, Kir\'{a}ly et al.~\cite{Kiraly:2012:icml,Kiraly:2015:jmlr} treated matrix completion as an algebraic problem and proposed deterministic conditions to decide whether a particular entry of a \emph{generic} matrix can be restored. Pimentel{-}Alarc{\'{o}}n et al.~\cite{daniel:2016:jstsp} built deterministic sampling conditions for ensuring that, \emph{almost surely}, there are only finitely many matrices that agree with the observed entries. However, strictly speaking, those conditions ensure only the recoverability of a special kind of matrices, but they cannot guarantee the identifiability of an arbitrary $L_0$ for sure. This gap is indeed striking, as the data matrices arising from modern applications are often of complicate structures and unnecessary to be generic. Moreover, the sampling conditions given in~\cite{Kiraly:2012:icml,Kiraly:2015:jmlr,daniel:2016:jstsp} are not so interpretable and thus not easy to use while applying to the other related problems such as \emph{matrix recovery} (which is matrix completion with $\Omega$ being unknown)~\cite{Candes:2009:JournalACM}.

To break through the limits of random sampling, we propose in this work two deterministic conditions, \emph{isomeric condition}~\cite{liu:nips:2017} and \emph{relative well-conditionedness}, for guaranteeing an \emph{arbitrary} matrix to be recoverable from a sampling of its entries. The isomeric condition is a mixed concept that combines together the rank and coherence of $L_0$ with the locations and amount of the observed entries. In general, isomerism (noun of isomeric) ensures that the \emph{sampled submatrices} (see Section \ref{sec:notation}) are not \emph{rank deficient}\footnote{In this paper, rank deficiency means that a submatrix does not have the largest possible rank. Specifically, suppose that $M'$ is a submatrix of some matrix $M$, then $M'$ is rank deficient iff (i.e., if and only if) $\rank{M'}<\rank{M}$. Note here that a submatrix is rank deficient does not necessarily mean that the submatrix does not have full rank, and a submatrix of full rank could be rank deficient.}. Remarkably, it is provable that isomerism is \emph{necessary} for the identifiability of $L_0$: Whenever the isomeric condition is violated, there exist infinity many matrices that can fit the observed entries not worse than $L_0$ does. Hence, logically speaking, the conditions given in~\cite{Kiraly:2012:icml,Kiraly:2015:jmlr,daniel:2016:jstsp} should suffice to ensure isomerism. While necessary, unfortunately isomerism does not suffice to guarantee the identifiability of $L_0$ in a deterministic fashion. This is because isomerism does not exclude the unidentifiable cases where the sampled submatrices are severely ill-conditioned. To compensate this weakness, we further propose the so-called \emph{relative well-conditionedness}, which encourages the smallest singular values of the sampled submatrices to be away from 0.

Equipped with these new tools, isomerism and relative well-conditionedness, we prove a set of theorems pertaining to \emph{missing data recovery}~\cite{Zhang06} and matrix completion. In particular, we prove that the exact solutions that identify the target matrix $L_0$ are strict local minima to the commonly used bilinear programs. Although theoretically sound, the classic bilinear programs suffer from a weakness that the rank of $L_0$ has to be known. To fix this flaw, we further consider a method termed \emph{isomeric dictionary pursuit} (IsoDP), the formula of which can be derived from Schatten quasi-norm minimization~\cite{rahul:jlmr:2010}, and we show that IsoDP is superior to the traditional bilinear programs. In summary, the main contribution of this work is to establish deterministic sampling conditions for ensuring the success in completing arbitrary matrices from a subset of the matrix entries, producing some theoretical results useful for understanding the completion regimes of arbitrary missing data patterns.
\section{Summary of Main Notations}\label{sec:notation}
Capital and lowercase letters are used to represent (real-valued) matrices and vectors, respectively, except that some lowercase letters, such as $i,j,k,m,n,l,p,q,r,s$ and $t$, are used to denote integers. For a matrix $M$, $[M]_{ij}$ is the $(i,j)$th entry of $M$, $[M]_{i,:}$ is its $i$th row, and $[M]_{:,j}$ is its $j$th column. Let $\omega_1=\{i_1,i_2,\cdots,i_k\}$ and $\omega_2=\{j_1,j_2,\cdots,j_s\}$ be two 1D sampling sets. Then $[M]_{\omega_1,:}$ denotes the submatrix of $M$ obtained by selecting the rows with indices $i_1,i_2,\cdots,i_k$, $[M]_{:,\omega_2}$ is the submatrix constructed by choosing the columns at $j_1,j_2,\cdots,j_s$, and similarly for $[M]_{\omega_1,\omega_2}$. For a 2D sampling set $\Omega\subseteq{}\{1,\cdots,m\} \times\{1,\cdots,n\}$, we imagine it as a sparse matrix and define its ``rows'', ``columns'' and ``transpose'' as follows: the $i$th row $\Omega_i = \{j_1 | (i_1,j_1)\in\Omega, i_1 = i\}$, the $j$th column $\Omega^j = \{i_1 | (i_1,j_1)\in\Omega, j_1 = j\}$, and the transpose $\Omega^T = \{(j_1,i_1) | (i_1,j_1)\in\Omega\}$. These notations are important for understanding the proposed conditions. For the ease of presentation, we shall call $[M]_{\omega,:}$ as a \emph{sampled submatrix} of $M$ (see Figure~\ref{fig:sub}), where $\omega$ is a 1D sampling set.
\begin{figure}
\begin{center}
\includegraphics[width=0.4\textwidth]{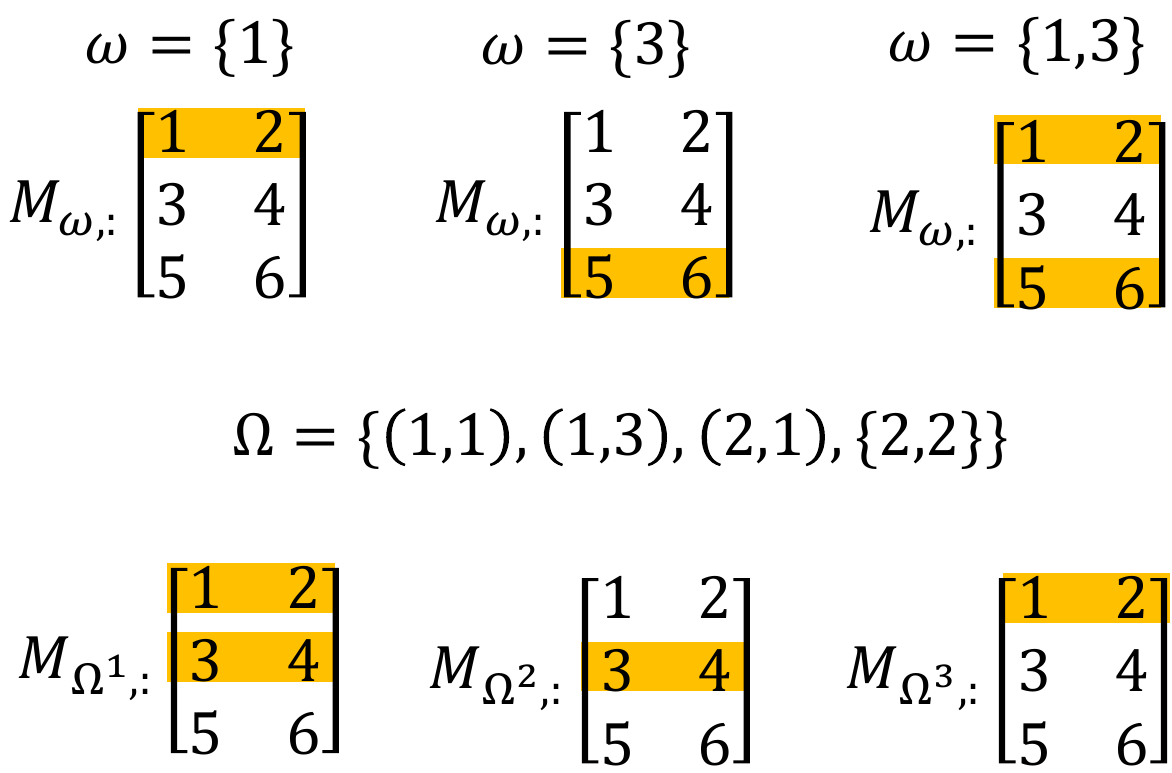}\vspace{-0.15in}
\caption{Illustrations of the sampled submatrices.}\label{fig:sub}\vspace{-0.25in}
\end{center}
\end{figure}

Three types of matrix norms are used in this paper: 1) the operator norm or 2-norm denoted by $\|M\|$, 2) the Frobenius norm denoted by $\|M\|_F$ and 3) the nuclear norm denoted by $\|M\|_*$. The only used vector norm is the $\ell_2$ norm, which is denoted by $\|\cdot\|_2$. Particularly, the symbol $|\cdot|$ is reserved for the cardinality of a set.

The special symbol $(\cdot)^+$ is reserved to denote the Moore-Penrose pseudo-inverse of a matrix. More precisely, for a matrix $M$ with SVD\footnote{In this paper, SVD always refers to skinny SVD. For a rank-$r$ matrix $M\in\mathbb{R}^{m\times{}n}$, its SVD is of the form $U_M\Sigma_MV_M^T$, where $U_M\in\Re^{m\times{}r},\Sigma_M\in\Re^{r\times{}r}$ and $V_M\in\Re^{n\times{}r}$.} $M=U_M\Sigma_MV_M^T$, its pseudo-inverse is given by $M^+=V_M\Sigma_M^{-1}U_M^T$. For convenience, we adopt the conventions of using $\mathrm{span}\{M\}$ to denote the linear space spanned by the columns of a matrix $M$, using $y\in\mathrm{span}\{M\}$ to denote that a vector $y$ belongs to the space $\mathrm{span}\{M\}$, and using $Y\in\mathrm{span}\{M\}$ to denote that all the column vectors of a matrix $Y$ belong to $\mathrm{span}\{M\}$.
\section{Identifiability Conditions}\label{sec:setting}
In this section, we introduce the so-called \emph{isomeric condition}~\cite{liu:nips:2017} and \emph{relative well-conditionedness}.
\subsection{Isomeric Condition}\label{sec:setting:iso}
For the ease of understanding, we shall begin with a concept called \emph{$k$-isomerism} (or \emph{$k$-isomeric} in adjective form), which can be regarded as an extension of low-rankness.
\begin{defn}[$k$-isomeric]\label{def:iso:k}
A matrix $M\in\Re^{m\times{}l}$ is called $k$-isomeric iff any $k$ rows of $M$ can linearly represent all rows in $M$. That is,
\begin{align*}
&\rank{[M]_{\omega,:}} = \rank{M}, \forall{}\omega\subseteq\{1,\cdots,m\}, |\omega| = k,
\end{align*}
where $|\cdot|$ is the cardinality of a sampling set and $[M]_{\omega,:}\in\mathbb{R}^{|\omega|\times{}l}$ is called a ``sampled submatrix'' of $M$.
\end{defn}
In short, a matrix $M$ is $k$-isomeric means that the sampled submatrix $[M]_{\omega,:}$ (with $|\omega|=k$) is not rank deficient\footnote{Here, the largest possible rank is $\rank{M}$. So $\rank{[M]_{\omega,:}} = \rank{M}$ gives that the submatrix $[M]_{\omega,:}$ is not rank deficient.}. According to the above definition, $k$-isomerism has a nice property; that is, suppose $M$ is $k_1$-isomeric, then $M$ is also $k_2$-isomeric for any $k_2\geq{}k_1$. So, to verify whether a matrix $M$ is $k$-isomeric with unknown $k$, one just needs to find the smallest $\bar{k}$ such that $M$ is $\bar{k}$-isomeric.

Generally, $k$-isomerism is somewhat similar to \emph{Spark}~\cite{Donoho:spark:2003}, which defines the smallest linearly dependent subset of the rows of a matrix. For a matrix $M$ to be $k$-isomeric, it is necessary that $\rank{M}\leq{}k$, not sufficient. In fact, $k$-isomerism is also somehow related to the concept of \emph{coherence}~\cite{Candes:2009:math,liu:tsp:2016}. For a rank-$r$ matrix $M\in\mathbb{R}^{m\times{}n}$ with SVD $U_M\Sigma_MV_M^T$, its coherence is denoted as $\mu(M)$ and given by
\begin{align*}
\mu(M)= \max(\max_{1\leq{}i\leq{}m}\frac{m}{r}\|[U_M]_{i,:}\|_F^2, \max_{1\leq{}j\leq{}n}\frac{n}{r}\|[V_M]_{j,:}\|_F^2).
\end{align*}
When the coherence of a matrix $M\in\Re^{m\times{}l}$ is not too high, $M$ could be $k$-isomeric with a small $k$, e.g., $k=\rank{M}$. Whenever the coherence of $M$ is very high, one may need a large $k$ to satisfy the $k$-isomeric property. For example, consider an extreme case where $M$ is a rank-1 matrix with one row being 1 and everywhere else being 0. In this case, we need $k=m$ to ensure that $M$ is $k$-isomeric. However, the connection between isomerism and coherence is not indestructible. A counterexample is the Hadamard matrix with $2^m$ rows and 2 columns. In this case, the matrix has an optimal coherence of 1, but the matrix is not $k$-isomeric for any $k\leq{}2^{m-1}$.

While Definition~\ref{def:iso:k} involves all 1D sampling sets of cardinality $k$, we often need the isomeric property to be associated with a certain 2D sampling set $\Omega$. To this end, we define below a concept called \emph{$\Omega$-isomerism} (or \emph{$\Omega$-isomeric}).
\begin{defn}[$\Omega$-isomeric]\label{def:iso:omg}
Let $M\in\Re^{m\times{}l}$ and $\Omega\subseteq\{1,\cdots,$ $m\}\times\{1,\cdots,n\}$. Suppose that $\Omega^j\neq\emptyset$ (empty set), $\forall{}1\leq{}j\leq{}n$. Then the matrix $M$ is called $\Omega$-isomeric iff
\begin{align*}
&\rank{[M]_{\Omega^j,:}} = \rank{M}, \forall{}j = 1,\cdots,n.
\end{align*}
Note here that $\Omega^j$ (i.e., $j$th column of $\Omega$) is a 1D sampling set and $l\neq{}n$ is allowed.
\end{defn}
Similar to $k$-isomerism, $\Omega$-isomerism also assumes that the sampled submatrices, $\{[M]_{\Omega^j,:}\}_{j=1}^n$, are not rank deficient. The main difference is that $\Omega$-isomerism requires the rank of $M$ to be preserved by the submatrices sampled according to a \emph{specific} sampling set $\Omega$, and $k$-isomerism assumes that \emph{every} submatrix consisting of $k$ rows of $M$ has the same rank as $M$. Hence, $\Omega$-isomerism is less strict than $k$-isomerism. More precisely, provided that $|\Omega^j|\geq{}k,\forall{}1\leq{}j\leq{}n$, a matrix $M$ is $k$-isomeric ensures that $M$ is $\Omega$-isomeric as well, but not vice versa. In the extreme case where $M$ is nonzero at only one row, interestingly, $M$ can be $\Omega$-isomeric as long as the locations of the nonzero entries are included in $\Omega$. For example, the following rank-1 matrix $M$ is not 1-isomeric but still $\Omega$-isomeric for some $\Omega$ with $|\Omega^j|=1,\forall{}1\leq{}j\leq{}n$:
\begin{align*}
\Omega = \{(1,1),(1,2), (1,3)\} \textrm{ and } M =\left[\begin{array}{cc}
1 &1\\
0&0\\
0&0
\end{array}\right],
\end{align*}
where it is configured that $m=n=3$ and $l=2$.

With the notation of $\Omega^T = \{(j_1,i_1) | (i_1,j_1)\in\Omega\}$, the isomeric property can be also defined on the column vectors of a matrix, as shown in the following definition.
\begin{defn}[$\Omega/\Omega^T$-isomeric]\label{def:iso:omgt}
Let $M\in\Re^{m\times{}n}$ and $\Omega\subseteq\{1,\cdots,m\}\times\{1,\cdots,n\}$. Suppose $\Omega_i\neq\emptyset$ and $\Omega^j\neq\emptyset$, $\forall{}i,j$. Then the matrix $M$ is called $\Omega/\Omega^T$-isomeric iff $M$ is $\Omega$-isomeric and $M^T$ is $\Omega^T$-isomeric as well.
\end{defn}
To solve Problem~\ref{pb:mc} without the assumption of missing at random, as will be shown later, it is necessary to assume that $L_0$ is $\Omega/\Omega^T$-isomeric. This condition has excluded the unidentifiable cases where any rows or columns of $L_0$ are wholly missing. Moreover, $\Omega/\Omega^T$-isomerism has partially considered the cases where $L_0$ is of high coherence: For the extreme case where $L_0$ is 1 at only one entry and 0 everywhere else, $L_0$ cannot be $\Omega/\Omega^T$-isomeric unless the index of the nonzero element is included in $\Omega$. In general, there are numerous reasons for the target matrix $L_0$ to be isomeric. For example, the standard assumptions of low-rankness, incoherence and uniform sampling are indeed sufficient to ensure isomerism, not necessary.
\begin{theo}\label{thm:iso}
Let $L_0\in\Re^{m\times{}n}$ and $\Omega\subseteq\{1,\cdots,m\}\times\{1,\cdots,$ $n\}$. Denote $n_1 = \max(m,n)$, $n_2=\min(m,n)$, $\mu_0=\mu(L_0)$ and $r_0=\rank{L_0}$. Suppose that $\Omega$ is a set sampled uniformly at random, namely $\mathrm{Pr}((i,j)\in\Omega)=\rho_0$ and $\mathrm{Pr}((i,j)\notin\Omega)=1-\rho_0$. If $\rho_0>c\mu_0r_0(\log{n_1})/n_2$ for some numerical constant $c$ then, with probability at least $1-n_1^{-10}$, $L_0$ is $\Omega/\Omega^T$-isomeric.
\end{theo}

Notice, that the isomeric condition can be also proven by discarding the uniform sampling assumption and accessing only the concept of coherence (see Theorem~\ref{thm:iso:rcn}). Furthermore, the isomeric condition could be even obeyed in the case of high coherence. For example,
\begin{align}\label{eq:example:1}
\hspace{-0.03in}\Omega \hspace{-0.03in}= \hspace{-0.03in}\{(1,1),\hspace{-0.03in} (1,2), \hspace{-0.03in}(1,3),\hspace{-0.03in} (2,1), \hspace{-0.03in}(3, 1)\} \textrm{ and } L_0 \hspace{-0.03in}=\hspace{-0.03in}\setlength\arraycolsep{0.1cm}\left[\hspace{-0.03in}\begin{array}{ccc}
1 &0&0\\
0&0&0\\
0&0&0
\end{array}\hspace{-0.03in}\right]\hspace{-0.03in},
\end{align}
where $L_0$ is not incoherent and the sampling is not uniform either, but it can be verified that $L_0$ is $\Omega/\Omega^T$-isomeric. In fact, the isomeric condition is \emph{necessary} for the identifiability of $L_0$, as shown in the following theorem.
\begin{theo}\label{thm:iso:necessary}
Let $L_0\in\Re^{m\times{}n}$ and $\Omega\subseteq\{1,\cdots,m\}\times\{1,\cdots,$ $n\}$. If either $L_0$ is not $\Omega$-isomeric or $L_0^T$ is not $\Omega^T$-isomeric then there exist infinity many matrices (denoted as $L\in\Re^{m\times{}n}$) that fit the observed entries not worse than $L_0$ does:
\begin{align*}
L\neq{}L_0,\textrm{ } \rank{L}\leq\rank{L_0},\textrm{ }[L]_{ij} = [L_0]_{ij},\forall{}(i,j)\in\Omega.
\end{align*}
\end{theo}
In other words, for any partial matrix $M'$ with sampling set $\Omega$, if there exists a completion $M$ that is not $\Omega/\Omega^T$-isomeric, then there are infinity many completions that are different from $M$ and have a rank not greater than that of $M$. In other words, isomerism is also necessary for the so-called \emph{finitely completable property} explored in~\cite{Kiraly:2012:icml,Kiraly:2015:jmlr,daniel:2016:jstsp}. As a consequence, logically speaking, the deterministic sampling conditions established in~\cite{Kiraly:2012:icml,Kiraly:2015:jmlr,daniel:2016:jstsp} should suffice to ensure isomerism. The above theorem illustrates that the isomeric condition is indeed necessary for the identifiability of the completions to any partial matrices, no matter how the observed entries are chosen.
\subsection{Relative Well-Conditionedness}
While necessary, the isomeric condition is unfortunately unable to guarantee the identifiability of $L_0$ for sure. More concretely, consider the following example:
\begin{align}\label{eq:example:2}
\Omega = \{(1, 1), (2, 2)\} \textrm{ and } L_0 =\left[\begin{array}{cc}
1 &\frac{10}{9}\\
\frac{9}{10} &1
\end{array}\right].
\end{align}
It can be verified that $L_0$ is $\Omega/\Omega^T$-isomeric. However, there still exist infinitely many rank-1 completions different than $L_0$, e.g., $L_*=[1 ,1; 1, 1]$, which is a matrix of all ones. For this particular example, $L_*$ is the optimal rank-1 completion in the sense of coherence. In general, isomerism is only a condition for the sampled submatrices to be not rank deficient, but there is no guarantee that the sampled submatrices are well-conditioned. To compensate this weakness, we further propose an additional hypothesis called \emph{relative well-conditionedness}, which encourages the smallest singular value of the sampled submatrices to be far from 0.

Again, we shall begin with a simple concept called \emph{$\omega$-relative condition number}, with $\omega$ being a 1D sampling set.
\begin{defn}[$\omega$-relative condition number]\label{def:rcn:1}
Let $M\in\Re^{m\times{}l}$ and $\omega\subseteq\{1,\cdots,m\}$. Suppose that $[M]_{\omega,:}\neq0$. Then the $\omega$-relative condition number of the matrix $M$ is denoted as $\gamma_{\omega}(M)$ and given by
\begin{align*}
\gamma_{\omega}(M) = 1/\|M([M]_{\omega,:})^+\|^2,
\end{align*}
where $(\cdot)^+$ and $\|\cdot\|$ are the pseudo-inverse and operator norm of a matrix, respectively.
\end{defn}
Regarding the bound of the $\omega$-relative condition number $\gamma_{\omega}(M)$, simple calculations yield
\begin{align*}
\sigma_{min}^2/\|M\|^2\leq\gamma_{\omega}(M)\leq1,
\end{align*}
where $\sigma_{min}$ is the smallest singular value of $[M]_{\omega,:}$. Hence, the sampled submatrix $[M]_{\omega,:}$ has a large minimum singular value is sufficient for ensuring that $\gamma_{\omega}(M)$ is large, not necessary. Roughly, the value of $\gamma_{\omega}(M)$ measures how much information of a matrix $M$ is contained in the sampled submatrix $[M]_{\omega,:}$. The more information $[M]_{\omega,:}$ contains, the larger $\gamma_{\omega}(M)$ is (this will be more clear later). For example, $\gamma_{\omega}(M)=1$ whenever $\omega=\{1,\cdots,m\}$. The concept of $\omega$-relative condition number can be extended to the case of 2D sampling sets, as shown below.
\begin{defn}[$\Omega$-relative condition number]\label{def:rcn:2}
Let $M\in\Re^{m\times{}l}$ and $\Omega\subseteq\{1,\cdots,m\}\times\{1,\cdots,n\}$. Suppose that $[M]_{\Omega^j,:}\neq0$, $\forall{}1\leq{}j\leq{}n$. Then the $\Omega$-relative condition number of $M$ is denoted as $\gamma_{\Omega}(M)$ and given by
\begin{align*}
\gamma_{\Omega}(M) = \min_{1\leq{}j\leq{}n}\gamma_{\Omega^j}(M),
\end{align*}
where $\Omega^j$ is a 1D sampling set corresponding to the $j$th column of $\Omega$. Again, note here that $l\neq{}n$ is allowed.
\end{defn}
Using the notation of $\Omega^T$, we can define the concept of $\Omega/\Omega^T$-relative condition number as in the following.
\begin{defn}[$\Omega/\Omega^T$-relative condition number]\label{def:rcn:3}
Let $M\in\Re^{m\times{}n}$ and $\Omega\subseteq\{1,\cdots,m\}\times\{1,\cdots,n\}$. Suppose that $[M]_{\Omega^j,:}\neq0$ and $[M]_{:,\Omega_i}\neq0$, $\forall{}1\leq{}i\leq{}m,1\leq{}j\leq{}n$. Then the $\Omega/\Omega^T$-relative condition number of $M$ is denoted as $\gamma_{\Omega,\Omega^T}(M)$ and given by
\begin{align*}
\gamma_{\Omega,\Omega^T}(M) = \min(\gamma_{\Omega}(M), \gamma_{\Omega^T}(M^T)).
\end{align*}
\end{defn}

To make sure that an arbitrary matrix $L_0$ is recoverable from a subset of the matrix entries, we need to assume that $\gamma_{\Omega,\Omega^T}(L_0)$ is reasonably large; this is the so-called \emph{relative well-conditionedness}. Under the standard settings of uniform sampling and incoherence, we have the following theorem to bound $\gamma_{\Omega,\Omega^T}(L_0)$.
\begin{theo}\label{thm:rcn:bound}
Let $L_0\in\Re^{m\times{}n}$ and $\Omega\subseteq\{1,\cdots,m\}\times\{1,\cdots,$ $n\}$. Denote $n_1 = \max(m,n)$, $n_2=\min(m,n)$, $\mu_0=\mu(L_0)$ and $r_0=\rank{L_0}$. Suppose that $\Omega$ is a set sampled uniformly at random, namely $\mathrm{Pr}((i,j)\in\Omega)=\rho_0$ and $\mathrm{Pr}((i,j)\notin\Omega)=1-\rho_0$. For any $\alpha>1$, if $\rho_0>\alpha{}c\mu_0r_0(\log{n_1})/n_2$ for some numerical constant $c$ then, with probability at least $1-n_1^{-10}$, $\gamma_{\Omega,\Omega^T}(L_0)>(1-1/\sqrt{\alpha})\rho_0$.
\end{theo}
The above theorem illustrates that, under the setting of uniform sampling \emph{plus} incoherence, the relative condition number approximately corresponds to the fraction of the observed entries. Actually, the relative condition number can be bounded from below without the assumption of uniform sampling.
\begin{theo}\label{thm:iso:rcn}
Let $L_0\in\Re^{m\times{}n}$ and $\Omega\subseteq\{1,\cdots,m\}\times\{1,\cdots,n\}$. Denote $\mu_0=\mu(L_0)$ and $r_0=\rank{L_0}$. Denote by $\rho$ the smallest fraction of the observed entries in each column and row of $L_0$; namely,
\begin{align*}
\rho = \min(\min_{1\leq{}i\leq{}m}\frac{|\Omega_{i}|}{n}, \min_{1\leq{}j\leq{}n}\frac{|\Omega^{j}|}{m}).
\end{align*}
For any $0\leq\alpha<1$, if $\rho>1-(1-\alpha)/(\mu_0r_0)$ then the matrix $L_0$ is $\Omega/\Omega^T$-isomeric and $\gamma_{\Omega,\Omega^T}(L_0)>\alpha$.
\end{theo}
It is worth noting that the relative condition number could be large even if the coherence of $L_0$ is extremely high. For the example shown in~\eqref{eq:example:1}, it can be calculated that $\gamma_{\Omega,\Omega^T}(L_0)=1$.
\section{Theories and Methods}\label{sec:mainbody}
In this section, we shall prove some theorems pertaining to matrix completion as well as missing data recovery. In addition, we suggest a method termed IsoDP for matrix completion, which possesses some remarkable features that we miss in the traditional bilinear programs.
\subsection{Missing Data Recovery}\label{sec:clue}
Before exploring the matrix completion problem, we would like to consider a missing data recovery problem studied by~\cite{Zhang06}, which is described as follows: Let $y_0\in\Re^m$ be a data vector drawn form some low-dimensional subspace, denoted as $y_0\in\mathcal{S}_0\subset\Re^m$. Suppose that $y_0$ contains some available observations in $y_b\in\Re^k$ and some missing entries in $y_u\in\Re^{m-k}$. Namely, after a permutation,
\begin{align}\label{eq:y}
y_0 = \left[\begin{array}{c}
y_b\\
y_u\\
\end{array}\right], y_b\in\Re^k, y_u\in\Re^{m-k}.
\end{align}
Given the observations in $y_b$, we seek to restore the unseen entries in $y_u$. To do this, we consider the prevalent idea that represents a data vector as a linear combination of the bases in a given dictionary:
\begin{align}\label{eq:ax}
y_0 = Ax_0,
\end{align}
where $A\in\Re^{m\times{}p}$ is a dictionary constructed in advance and $x_0\in\Re^{p}$ is the representation of $y_0$. Utilizing the same permutation used in~\eqref{eq:y}, we can partition the rows of $A$ into two parts according to the locations of the observed and missing entries:
\begin{align}\label{eq:A}
A = \left[\begin{array}{c}
A_b\\
A_u\\
\end{array}\right], A_b\in\Re^{k\times{}p}, A_u\in\Re^{(m-k)\times{}p}.
\end{align}
In this way, the equation in~\eqref{eq:ax} gives that
\begin{align*}
y_b = A_bx_0\quad\text{and}\quad{}y_u = A_ux_0.
\end{align*}
As we now can see, the unseen data $y_u$ is exactly restored, as long as the representation $x_0$ is retrieved by only accessing the available observations in $y_b$. In general cases, there are infinitely many representations that satisfy $y_0 = Ax_0$, e.g., $x_0=A^+y_0$, where $(\cdot)^+$ is the pseudo-inverse of a matrix. Since $A^+y_0$ is the representation of minimal $\ell_2$ norm, we revisit the traditional $\ell_2$ program:
\begin{align}\label{eq:l2}
\min_{x} \frac{1}{2}\norm{x}_2^2,\quad\textrm{s.t.}\quad{}y_b = A_bx,
\end{align}
where $\|\cdot\|_2$ is the $\ell_2$ norm of a vector. The above problem has a closed-form solution given by $A_b^+y_b$. Under some verifiable conditions, the above $\ell_2$ program is indeed \emph{consistently successful} in a sense as in the following: For any $y_0\in\mathcal{S}_0$ with an arbitrary partition $y_0=[y_b;y_u]$ (i.e., arbitrarily missing), the desired representation $x_0=A^+y_0$ is the unique minimizer to the problem in~\eqref{eq:l2}. That is, the unseen data $y_u$ is exactly recovered by firstly computing $x_*=A_b^+y_b$ and then calculating $y_u=A_ux_*$.
\begin{theo}\label{thm:l2} Let $y_0=[y_b;y_u]\in\Re^m$ be an authentic sample drawn from some low-dimensional subspace $\mathcal{S}_0$. Denote by $k$ the number of available observations in $y_b$. Then the convex program~\eqref{eq:l2} is consistently successful, as long as $\mathcal{S}_0\subseteq\mathrm{span}\{A\}$ and the given dictionary $A$ is $k$-isomeric.
\end{theo}

The above theorem says that, in order to recover an $m$-dimensional vector sampled from some subspace determined by a given $k$-isomeric dictionary $A$, one only needs to see $k$ entries of the vector.
\subsection{Convex Matrix Completion}
Low rank matrix completion concerns the problem of seeking a matrix that not only attains the lowest rank but also satisfies the constraints given by the observed entries:
\begin{eqnarray*}
\min_{L} \rank{L},\quad\textrm{s.t.}\quad{}[L]_{ij} = [L_0]_{ij},\forall{}(i,j)\in\Omega.
\end{eqnarray*}
Unfortunately, this idea is of little practical because the problem above is essentially NP-hard and cannot be solved in polynomial time~\cite{Chistov:1984}. To achieve practical matrix completion, Cand{\`e}s and Recht~\cite{Candes:2009:math,Recht2008} suggested an alternative that minimizes instead the nuclear norm; namely,
\begin{eqnarray}\label{eq:numin}
\min_{L} \|L\|_*,\quad\textrm{s.t.}\quad{}[L]_{ij} = [L_0]_{ij},\forall{}(i,j)\in\Omega,
\end{eqnarray}
where $\|\cdot\|_*$ denotes the nuclear norm, i.e., the sum of the singular values of a matrix. Under the context of uniform sampling, it has been proved that the above convex program succeeds in recovering the target $L_0$.

Although its theory is built upon the assumption of missing at random, as observed widely in the literatures, the convex program~\eqref{eq:numin} actually works even when the locations of the missing entries are distributed in a correlated and nonuniform fashion. This phenomenon could be explained by the following theorem, which states that the solution to the problem in~\eqref{eq:numin} is \emph{unique} and \emph{exact}, provided that the isomeric condition is obeyed and the relative condition number of $L_0$ is large enough.
\begin{theo}\label{thm:convex}
Let $L_0\in\Re^{m\times{}n}$ and $\Omega\subseteq\{1,\cdots,m\}\times\{1,\cdots,n\}$. If $L_0$ is $\Omega/\Omega^T$-isomeric and $\gamma_{\Omega,\Omega^T}(L_0)>0.75$ then $L_0$ is the unique minimizer to the problem in~\eqref{eq:numin}.
\end{theo}

Roughly speaking, the assumption $\gamma_{\Omega,\Omega^T}(L_0)>0.75$ requires that more than three quarters of the information in $L_0$ is observed. Such an assumption is seemingly restrictive but technically difficult to reduce in general cases.
\subsection{Nonconvex Matrix Completion}\label{sec:mainres}
The problem of missing data recovery is closely related to matrix completion, which is actually to restore the missing entries in multiple data vectors simultaneously. Hence, we would transfer the spirits of the $\ell_2$ program~\eqref{eq:l2} to the case of matrix completion. Following~\eqref{eq:l2}, one may consider Frobenius norm minimization for matrix completion:
\begin{align}\label{eq:fnorm}
\min_{X} \frac{1}{2}\norm{X}_F^2,\textrm{ s.t. }[AX]_{ij} = [L_0]_{ij},\forall{}(i,j)\in\Omega,
\end{align}
where $A\in\Re^{m\times{}p}$ is a dictionary matrix assumed to be given. Similar to~\eqref{eq:l2}, the convex program~\eqref{eq:fnorm} can also exactly recover the desired representation matrix $A^+L_0$, as shown in the theorem below.
\begin{theo}\label{thm:fnorm}
Let $L_0\in\Re^{m\times{}n}$ and $\Omega\subseteq\{1,\cdots,m\}\times\{1,\cdots,n\}$. Provided that $L_0\in\mathrm{span}\{A\}$ and the given dictionary $A$ is $\Omega$-isomeric, the desired representation $X_0=A^+L_0$ is the unique minimizer to the problem in~\eqref{eq:fnorm}.
\end{theo}

Theorem~\ref{thm:fnorm} tells us that, in general, even when the locations of the missing entries are placed arbitrarily, the target $L_0$ is restored as long as we have a proper dictionary $A$. This motivates us to consider the commonly used bilinear program that seeks both $A$ and $X$ simultaneously:
\begin{align}\label{eq:isodp:f}
\hspace{-0.05in}\min_{A,X}\frac{1}{2} (\norm{A}_F^2\hspace{-0.02in}+\hspace{-0.02in} \norm{X}_F^2),\textrm{ s.t. }[AX]_{ij} \hspace{-0.02in}= \hspace{-0.02in} [L_0]_{ij},\forall{}(i,j)\hspace{-0.02in}\in\hspace{-0.02in}\Omega,
\end{align}
where $A\in\Re^{m\times{}p}$ and $X\in\Re^{p\times{}n}$. The problem above is bilinear and therefore nonconvex. So, it would be hard to obtain a strong performance guarantee as done in the convex programs, e.g.,~\cite{Candes:2009:math,liu:tsp:2016}. What is more, the setup of deterministic sampling requires a deterministic recovery guarantee, the proof of which is much more difficult than a probabilistic guarantee. Interestingly, under the very mild condition of isomerism, the problem in~\eqref{eq:isodp:f} is proven to include the exact solutions that identify the target matrix $L_0$ as the critical points. Furthermore, when the relative condition number of $L_0$ is sufficiently large, the local optimality of the exact solutions is guaranteed surely.
\begin{theo}\label{thm:isodp:f}
Let $L_0\in\Re^{m\times{}n}$ and $\Omega\subseteq\{1,\cdots,m\}\times\{1,\cdots,n\}$. Denote the rank and the SVD of $L_0$ as $r_0$ and $U_0\Sigma_0V_0^T$, respectively. Define
\begin{align*}
&A_0 = U_0\Sigma_0^{\frac{1}{2}}Q^T, X_0= Q\Sigma_0^{\frac{1}{2}}V_0^T, \forall{}Q\in\Re^{p\times{}r_0}, Q^TQ = \Id.
\end{align*}
Then we have the following:
\begin{itemize}
\item[1.]If $L_0$ is $\Omega/\Omega^T$-isomeric then the exact solution, denoted as $(A_0, X_0)$, is a critical point to the problem in~\eqref{eq:isodp:f}.
\item[2.]If $L_0$ is $\Omega/\Omega^T$-isomeric, $\gamma_{\Omega,\Omega^T}(L_0)>0.5$ and $p=r_0$ then $(A_0, X_0)$ is a local minimum to the problem in~\eqref{eq:isodp:f}, and the local optimality is strict while ignoring the differences among the exact solutions that equally recover $L_0$.
\end{itemize}
\end{theo}

The condition of $\gamma_{\Omega,\Omega^T}(L_0)>0.5$, roughly, demands that more than half of the information in $L_0$ is observed. Unless some extra assumptions are imposed, this condition is not reducible, because counterexamples do exist when $\gamma_{\Omega,\Omega^T}(L_0)<0.5$. Consider a concrete case with
\begin{align}\label{eq:example:3}
\Omega = \{(1, 1), (2, 2)\} \textrm{ and } L_0 =\left[\begin{array}{cc}
1 &\sqrt{\alpha^2-1}\\
\frac{1}{\sqrt{\alpha^2-1}} &1
\end{array}\right],
\end{align}
where $\alpha>\sqrt{2}$. Then it can be verified that $L_0$ is $\Omega/\Omega^T$-isomeric. Via some calculations, we have (assume $p=r_0$)
\begin{align*}
&\gamma_{\Omega,\Omega^T}(L_0) = \min(1-\frac{1}{\alpha^2},\frac{1}{\alpha^2})=\frac{1}{\alpha^2} < 0.5,\\
&A_0 = \left[\begin{array}{c}
(\alpha^2-1)^{\frac{1}{4}}\\
\frac{1}{(\alpha^2-1)^{\frac{1}{4}}}\\
\end{array}\right]\textrm{ and } X_0 = \left[\frac{1}{(\alpha^2-1)^{\frac{1}{4}}}, (\alpha^2-1)^{\frac{1}{4}}\right].
\end{align*}
Now, construct
\begin{align*}
&A_{\epsilon} = \left[\begin{array}{c}
\frac{(\alpha^2-1)^{\frac{1}{4}}}{1+\epsilon}\\
1/(\alpha^2-1)^{\frac{1}{4}}\\
\end{array}\right]\textrm{ and } X_{\epsilon} = \left[\frac{1+\epsilon}{(\alpha^2-1)^{\frac{1}{4}}}, (\alpha^2-1)^{\frac{1}{4}}\right],
\end{align*}
where $\epsilon>0$. It is easy to see that $(A_{\epsilon},X_{\epsilon})$ is a feasible solution to~\eqref{eq:isodp:f}. However, as long as $0<\epsilon<\sqrt{\alpha^2-1}-1$, it can be verified that
\begin{align*}
\|A_{\epsilon}\|_F^2 + \|X_{\epsilon}\|_F^2 < \|A_0\|_F^2 + \|X_0\|_F^2,
\end{align*}
which implies that $(A_0,X_0)$ is not a local minimum to~\eqref{eq:isodp:f}. In fact, for the particular example shown in~\eqref{eq:example:3}, it can be proven that a global minimum to~\eqref{eq:isodp:f} is given by $(A_*=[1 ;1], X_*=[1,1])$, which cannot correctly reconstruct $L_0$.
\subsection{Isomeric Dictionary Pursuit}
Theorem~\ref{thm:isodp:f} illustrates that program~\eqref{eq:isodp:f} relies on the assumption of $p=\rank{L_0}$. This is consistent with the widely observed phenomenon that program~\eqref{eq:isodp:f} may not work well while the parameter $p$ is far from the true rank of $L_0$. To overcome this drawback, again, we recall Theorem~\ref{thm:fnorm}. Notice, that the $\Omega$-isomeric condition imposed on the dictionary matrix $A$ requires that
\begin{align*}
\rank{A}\leq|\Omega^j|,\forall{}j=1,\cdots,n.
\end{align*}
This, together with the condition of $L_0\in\mathrm{span}\{A\}$, motivates us to combine the formulation~\eqref{eq:fnorm} with the popular idea of nuclear norm minimization, resulting in a bilinear program termed IsoDP, which estimates both $A$ and $X$ by minimizing a mixture of the nuclear and Frobenius norms:
\begin{align}\label{eq:isodp}
\hspace{-0.05in}\min_{A,X}\norm{A}_*\hspace{-0.03in}+\hspace{-0.03in}\frac{1}{2}\norm{X}_F^2,\textrm{ s.t. }[AX]_{ij} \hspace{-0.03in}= \hspace{-0.03in} [L_0]_{ij},\hspace{-0.02in}\forall{}(i,j)\hspace{-0.02in}\in\hspace{-0.02in}\Omega,
\end{align}
where $A\in\Re^{m\times{}p}$ and $X\in\Re^{p\times{}n}$. The above formula can be also derived from the framework of Schatten quasi-norm minimization~\cite{rahul:jlmr:2010,Shang:2016:SAT,xu:2017:aai}. It has been proven in~\cite{Shang:2016:SAT,xu:2017:aai} that, for any rank-$r$ matrix $L\in\Re^{m\times{}n}$ with singular values $\sigma_1,\cdots,\sigma_r$, the following holds:
\begin{align}\label{eq:snorm}
\frac{1}{q}\|L\|_{q}^q = \min_{A,X}\frac{1}{q_1} \|A\|_{q_1}^{q_1} + \frac{1}{q_2}\|X\|_{q_2}^{q_2}, \textrm{ s.t. } AX = L,
\end{align}
as long as $p\geq{}r$ and $1/q = 1/q_1+1/q_2$ ($q,q_1,q_2>0$), where $\|L\|_q = (\sum_{i=1}^r\sigma_i^q)^{1/q}$ is the Schatten-$q$ norm. In that sense, the IsoDP program~\eqref{eq:isodp} is related to the following Schatten-$q$ quasi-norm minimization problem with $q = 2/3$:
\begin{align}\label{eq:stmin}
\min_{L} \frac{3}{2}\|L\|_{2/3}^{2/3} ,\quad\textrm{s.t.}\quad{}[L]_{ij} = [L_0]_{ij},\forall{}(i,j)\in\Omega.
\end{align}
Nevertheless, programs~\eqref{eq:stmin} and~\eqref{eq:isodp} are not equivalent to each other; this is obvious if $p<m$ (assume $m\leq{}n$). In fact, even when $p\geq{}m$, the conclusion~\eqref{eq:snorm} only implies that the global minima of~\eqref{eq:stmin} and~\eqref{eq:isodp} are equivalent, but their local minima and critical points could be different. More precisely, any local minimum to~\eqref{eq:stmin} certainly corresponds to a local minimum to~\eqref{eq:isodp}, but not vice versa\footnote{Suppose that $L_1$ is a local minimum to the problem in~\eqref{eq:stmin}. Let $(A_1,X_1) = \arg\min_{A,X} \norm{A}_*+0.5\norm{X}_F^2$, s.t. $AX=L_1$. Then $(A_1,X_1)$ has to be a local minimum to~\eqref{eq:isodp}. This can be proven by the method of reduction to absurdity. Assume that $(A_1,X_1)$ is not a local minimum to~\eqref{eq:isodp}. Then there exists some feasible solution, denoted as $(A_2, X_2)$, that is arbitrarily close to $(A_1, X_1)$ and satisfies $\norm{A_2}_*+0.5\norm{X_2}_F^2 < \norm{A_1}_*+0.5\norm{X_1}_F^2$. Taking $L_2=A_2X_2$, we have that $L_2$ is arbitrarily close to $L_1$ and $\frac{3}{2}\|L_2\|_{2/3}^{2/3}\leq\norm{A_2}_*+0.5\norm{X_2}_F^2 < \norm{A_1}_*+0.5\norm{X_1}_F^2=\frac{3}{2}\|L_1\|_{2/3}^{2/3}$, which contradicts the premise that $L_1$ is a local minimum to~\eqref{eq:stmin}. So, a local minimum to~\eqref{eq:stmin} also gives a local minimum to~\eqref{eq:isodp}. But the converge of this statement may not be true, and~\eqref{eq:isodp} might have more local minima than~\eqref{eq:stmin}.}. For the same reason, the bilinear program~\eqref{eq:isodp:f} is not equivalent to the convex program~\eqref{eq:numin}.

Regarding the recovery performance of the IsoDP program~\eqref{eq:isodp}, we establish the following theorem that reproduces Theorem~\ref{thm:isodp:f} without the assumption of $p=r_0$.
\begin{theo}\label{thm:isodp}
Let $L_0\in\Re^{m\times{}n}$ and $\Omega\subseteq\{1,\cdots,m\}\times\{1,\cdots,n\}$. Denote the rank and the SVD of $L_0$ as $r_0$ and $U_0\Sigma_0V_0^T$, respectively. Define
\begin{align*}
&A_0 = U_0\Sigma_0^{\frac{2}{3}}Q^T, X_0= Q\Sigma_0^{\frac{1}{3}}V_0^T,\forall{}Q\in\Re^{p\times{}r_0}, Q^TQ = \Id.
\end{align*}
Then we have the following:
\begin{itemize}
\item[1.]If $L_0$ is $\Omega/\Omega^T$-isomeric then the exact solution $(A_0, X_0)$ is a critical point to the problem in~\eqref{eq:isodp}.
\item[2.]If $L_0$ is $\Omega/\Omega^T$-isomeric and $\gamma_{\Omega,\Omega^T}(L_0)>0.5$ then $(A_0, X_0)$ is a local minimum to the problem in~\eqref{eq:isodp}, and the local optimality is strict while ignoring the differences among the exact solutions that equally recover $L_0$.
\end{itemize}
\end{theo}
Due to the advantages of the nuclear norm, the above theorem does not require the assumption of $p=\rank{L_0}$ any more. Empirically, unlike~\eqref{eq:isodp:f}, which exhibits superior performance only if $p$ is close to $\rank{L_0}$ and the initial solution is chosen carefully, IsoDP can work well by simply choosing $p=m$ and using $A=\Id$ as the initial solution.
\subsection{Optimization Algorithm}\label{sec:opt}
Considering the fact that the observations in reality are often contaminated by noise, we shall investigate instead the following bilinear program that can also approximately solve the problem in~\eqref{eq:isodp}:
\begin{align}\label{eq:isodp:noisy}
&\hspace{-0.05in}\min_{A,X} \lambda(\norm{A}_*\hspace{-0.02in}+\hspace{-0.02in}\frac{1}{2}\norm{X}_F^2)\hspace{-0.02in}+\hspace{-0.02in}\frac{1}{2}\sum_{(i,j)\in\Omega}([AX]_{ij}\hspace{-0.02in}-\hspace{-0.02in}[L_0]_{ij})^2,
\end{align}
where $A\in\Re^{m\times{}m}$ (i.e., $p=m$), $X\in\Re^{m\times{}n}$ and $\lambda>0$ is taken as a parameter.

The optimization problem in~\eqref{eq:isodp:noisy} can be solved by any of the many first-order methods established in the literatures. For the sake of simplicity, we choose to use the proximal methods by~\cite{proximal:2009:mp,Bolte2014}. Let $(A_t,X_t)$ be the solution estimated at the $t$th iteration. Define a function $g_t(\cdot)$ as
\begin{align*}
g_t(A) = \frac{1}{2}\sum_{(i,j)\in\Omega}([AX_{t+1}]_{ij}-[L_0]_{ij})^2.
\end{align*}
Then the solution to~\eqref{eq:isodp:noisy} is updated via iterating the following two procedures:
\begin{align}\label{eq:proximal}
&\hspace{-0.05in}X_{t+1}\hspace{-0.02in}= \hspace{-0.02in}\arg\min_{X} \frac{\lambda}{2}\|X\|_F^2\hspace{-0.02in}+\hspace{-0.02in}\frac{1}{2}\sum_{(i,j)\in\Omega}([A_tX]_{ij}-[L_0]_{ij})^2,\\\nonumber
&\hspace{-0.05in}A_{t+1}\hspace{-0.02in}=\hspace{-0.02in}  \arg\min_{A} \frac{\lambda}{\mu_t}\|A\|_*+\frac{1}{2}\|A - (A_t-\frac{\partial{}g_t(A_t)}{\mu_t})\|_F^2,
\end{align}
where $\mu_t>0$ is a penalty parameter and $\partial{}g_t(A_t)$ is the gradient of the function $g_t(A)$ at $A=A_t$. According to~\cite{proximal:2009:mp}, the penalty parameter $\mu_t$ could be set as $\mu_t = \|X_{t+1}\|^2$. The two optimization problems in~\eqref{eq:proximal} both have closed-form solutions. To be more precise, the $X$-subproblem is a least square regression problem:
\begin{align}\label{eq:x-sub}
[X_{t+1}]_{:,j} = (A_j^TA_j+\lambda\Id)^{-1}A_j^Ty_j, \forall{1\leq{}j\leq{}n},
\end{align}
where $A_j = [A_t]_{\Omega^j,:}$ and $y_j=[L_0]_{\Omega^j,j}$. The $A$-subproblem is solved by Singular Value Thresholding (SVT)~\cite{svt:cai:2008}:
\begin{align}\label{eq:a-sub}
A_{t+1}=U\mathcal{H}_{\lambda/\mu_t}(\Sigma)V^T,
\end{align}
where $U\Sigma{}V^T$ is the SVD of $A_t-\partial{}g_t(A_t)/\mu_t$ and $\mathcal{H}_{\lambda/\mu_t}(\cdot)$ denotes the shrinkage operator with parameter $\lambda/\mu_t$.

The whole optimization procedure is also summarized in Algorithm~\ref{alg1}. Without loss of generality, assume that $m\leq{}n$. Then the computational complexity of each iteration in Algorithm~\ref{alg1} is $O(m^2n)+O(m^3)$.
\begin{algorithm}[htb]
\caption{Solving problem~\eqref{eq:isodp:noisy} by alternating proximal}
\label{alg1}
\begin{algorithmic}[1]
\STATE \textbf{Input}: $\{[L_0]_{ij} |(i,j)\in\Omega\}$.
\STATE \textbf{Output}: the dictionary $A$ and the representation $X$.
\STATE \textbf{Initialization}: $A=\Id$.
\REPEAT
\STATE Update the representation matrix $X$ by~\eqref{eq:x-sub}.
\STATE Update the dictionary matrix $A$ by~\eqref{eq:a-sub}.
\UNTIL{convergence}
\end{algorithmic}
\end{algorithm}
\section{Mathematical Proofs}\label{sec:proof}
This section shows the detailed proofs of the theorems proposed in this work.
\subsection{Notations}
Besides of the notations presented in Section~\ref{sec:notation}, there are some other notations used throughout the proofs. Letters $U$, $V$, $\Omega$ and their variants (complements, subscripts, etc.) are reserved for left singular vectors, right singular vectors and support set, respectively. For convenience, we shall abuse the notation $U$ (resp. $V$) to denote the linear space spanned by the columns of $U$ (resp. $V$), i.e., the column space (resp. row space). The orthogonal projection onto the column space $U$, is denoted by $\mathcal{P}_U$ and given by $\mathcal{P}_U(M)=UU^TM$, and similarly for the row space $\mathcal{P}_V(M)=MVV^T$. Also, we denote by $\mathcal{P}_T$ the projection to the sum of the column space $U$ and the row space $V$, i.e., $\mathcal{P}_T(\cdot) = UU^T(\cdot)+(\cdot)VV^T-UU^T(\cdot)VV^T$. The same notation is also used to represent a subspace of matrices (i.e., the image of an operator), e.g., we say that $M\in\mathcal{P}_{U}$ for any matrix $M$ which satisfies $\mathcal{P}_{U}(M)=M$. The symbol $\mathcal{P}_{\Omega}$ denotes the orthogonal projection onto $\Omega$:
\begin{align*}
[\mathcal{P}_\Omega(M)]_{ij}=\left\{\begin{array}{cc}
[M]_{ij},&\text{if }(i,j)\in\Omega,\\
0, &\text{otherwise.}\\
\end{array}\right.
\end{align*}
Similarly, the symbol $\mathcal{P}_{\Omega}^{\bot}$ denotes the orthogonal projection onto the complement space of $\Omega$; that is, $\mathcal{P}_{\Omega}+\mathcal{P}_{\Omega}^{\bot}=\mathcal{I}$, where $\mathcal{I}$ is the identity operator.
\vspace{-0.1in}\subsection{Basic Lemmas}
While its definitions are associated with a certain matrix, the isomeric condition is actually characterizing some properties of a space, as shown in the lemma below.
\begin{lemm}\label{lem:basic:L02U}
Let $L_0\in\Re^{m\times{}n}$ and $\Omega\subseteq\{1,\cdots,m\}\times\{1,\cdots,n\}$. Denote the SVD of $L_0$ as $U_0\Sigma_0V_0^T$. Then we have:
\begin{itemize}
\item[1.] $L_0$ is $\Omega$-isomeric iff $U_0$ is $\Omega$-isomeric.
\item[2.] $L_0^T$ is $\Omega^T$-isomeric iff $V_0$ is $\Omega^T$-isomeric.
\end{itemize}
\end{lemm}
\begin{proof}
It can be manipulated that
\begin{align*}
[L_0]_{\Omega^j,:} = ([U_0]_{\Omega^j,:})\Sigma_0V_0^T, \forall{}j=1,\cdots, n.
\end{align*}
Since $\Sigma_0V_0^T$ is row-wisely full rank, we have
\begin{align*}
\rank{[L_0]_{\Omega^j,:}} = \rank{[U_0]_{\Omega^j,:}},\forall{}j=1,\cdots,n.
\end{align*}
As a consequence, $L_0$ is $\Omega$-isomeric is equivalent to $U_0$ is $\Omega$-isomeric. Similarly, the second claim is proven.
\end{proof}

The isomeric property is indeed subspace successive, as shown in the next lemma.
\begin{lemm}\label{lem:basic:subsucc}
Let $\Omega\subseteq\{1,\cdots,m\}\times\{1,\cdots,n\}$ and $U_0\in\Re^{m\times{}r}$ be the basis matrix of a subspace embedded in $\Re^m$. Suppose that $U$ is a subspace of $U_0$, i.e., $U = U_0U_0^TU$. If $U_0$ is $\Omega$-isomeric then $U$ is $\Omega$-isomeric as well.
\end{lemm}
\begin{proof}By $U = U_0U_0^TU$ and $U_0$ is $\Omega$-isomeric,
\begin{align*}
&\rank{[U]_{\Omega^j,:}} = \rank{([U_0]_{\Omega^j,:})U_0^TU}=\rank{U_0^TU}\\
&=\rank{U_0U_0^TU}=\rank{U}, \forall{}1\leq{}j\leq{}n.
\end{align*}
\end{proof}

The following lemma reveals the fact that the isomeric property is related to the invertibility of matrices.
\begin{lemm}\label{lem:basic:positive}
Let $\Omega\subseteq\{1,\cdots,m\}\times\{1,\cdots,n\}$ and $U_0\in\Re^{m\times{}r}$ be the basis matrix of a subspace of $\Re^m$. Denote by $u_i^T$ the $i$th row of $U_0$, i.e., $U_0 = [u_1^T;\cdots;u_m^T]$. Define $\delta_{ij}$ as
\begin{align}\label{eq:delta}
\delta_{ij}=\left\{\begin{array}{cc}
1,&\text{if }(i,j)\in\Omega,\\
0, &\text{otherwise.}\\
\end{array}\right.
\end{align}
Then the matrices, $\sum_{i=1}^{m}\delta_{ij}u_iu_i^T$, $\forall{}1\leq{}j\leq{}n$, are all invertible iff $U_0$ is $\Omega$-isomeric.
\end{lemm}
\begin{proof} Note that
\begin{align*}
&([U_0]_{\Omega^j,:})^T([U_0]_{\Omega^j,:})=\sum_{i=1}^{m}(\delta_{ij})^2u_iu_i^T=\sum_{i=1}^{m}\delta_{ij}u_iu_i^T.
\end{align*}
Now, it is easy to see that the matrix $\sum_{i=1}^{m}\delta_{ij}u_iu_i^T$ is invertible is equivalent to the matrix $([U_0]_{\Omega^j,:})^T([U_0]_{\Omega^j,:})$ is positive definite, which is further equivalent to $\rank{[U_0]_{\Omega^j,:}}=\rank{U_0}$, $\forall{}j=1,\cdots,n$.
\end{proof}

The following lemma gives some insights to the relative condition number.
\begin{lemm}\label{lem:basic:rcn} Let $M\in\Re^{m\times{}l}$ and $\omega\subseteq\{1,\cdots,m\}$. Define $\{\delta_i\}_{i=1}^m$ with $\delta_i = 1$ if $i\in\omega$ and 0 otherwise. Define a dialog matrix $D\in\mathbb{R}^{m\times{}m}$ as $D=\diag{\delta_1,\delta_2,\cdots,\delta_m}$. Denote the SVD of $M$ as $U\Sigma{}V^T$. If $\rank{[M]_{\omega,:}} = \rank{M}$ then
\begin{align*}
\gamma_{\omega}(M) = \sigma_{min},
\end{align*}
where $\sigma_{min}$ is the the smallest singular value (or eigenvalue) of the matrix $U^TDU$.
\end{lemm}
\begin{proof}
First note that $[M]_{\omega,:}$ can be equivalently written as $DU\Sigma{}V^T$. By the assumption of $\rank{[M]_{\omega,:}} = \rank{M}$, $DU$ is column-wisely full rank. Thus,
\begin{align*}
&M([M]_{\omega,:})^+ = U\Sigma{}V^T(DU\Sigma{}V^T)^+ = U\Sigma{}V^T(\Sigma{}V^T)^+(DU)^+\\
&=U(DU)^+=U(U^TDU)^{-1}U^TD,
\end{align*}
which gives that
\begin{align*}
&M([M]_{\omega,:})^+(M([M]_{\omega,:})^+)^T = U(U^TDU)^{-1}U^T.
\end{align*}
As a result, we have $\|M([M]_{\omega,:})^+\|^2 = 1/\sigma_{min}$, and thereby
\begin{align*}
\gamma_{\omega}(M) = 1/\|M([M]_{\omega,:})^+\|^2 = \sigma_{min}.
\end{align*}
\end{proof}

It has been proven in~\cite{siam_2010_minirank} that $\|L\|_*=\min_{A,X}\frac{1}{2}(\|A\|_F^2+\|X\|_F^2), \textrm{ s.t. }AX=L$. We have an analogous result, which has also been proven by~\cite{rahul:jlmr:2010,Shang:2016:SAT,xu:2017:aai}.
\begin{lemm}\label{lem:basic:ax} Let $L\in\Re^{m\times{}n}$ be a rank-$r$ matrix with $r\leq{}p$. Denote the SVD of $L$ as $U\Sigma{}V^T$. Then we have the following:
\begin{align*}
\frac{3}{2}\trace{\Sigma^{\frac{2}{3}}} = \min_{A\in\Re^{m\times{}p},X\in\Re^{p\times{}n}}\|A\|_*+\frac{1}{2}\|X\|_F^2, \textrm{ s.t. }AX = L,
\end{align*}
where $\trace{\cdot}$ is the trace of a square matrix.
\end{lemm}
\begin{proof}
Denote the singular values of $L$ as $\sigma_1\geq\cdots\geq\sigma_r>0$. We first consider the case that $\rank{A}=\rank{L}=r$. Since $AX=L$, the SVD of $A$ must have a form of $UQ\Sigma_AV_A^T$, where $Q$ is an orthogonal matrix of size $r\times{}r$ and $\Sigma_A = \diag{\alpha_1,\cdots,\alpha_r}$ with $\alpha_1\geq\cdots\geq\alpha_r>0$. Since $A^+L = \arg\min_{X} \|X\|_F^2, \textrm{ s.t. } AX=L$, we have
\begin{align*}
&\|A\|_* + \frac{1}{2}\|X\|_F^2 \geq \|A\|_* + \frac{1}{2}\|A^+L\|_F^2\\
&=\trace{\Sigma_A}+\frac{1}{2}\trace{\Sigma_A^{-1}Q^T\Sigma^2Q\Sigma_A^{-1}}.
\end{align*}
It can be proven that the eigenvalues of $\Sigma_A^{-1}Q^T\Sigma^2Q\Sigma_A^{-1}$ are given by $\{\sigma_i^2/\alpha_{\pi_i}^2\}_{i=1}^r$, where $\{\alpha_{\pi_i}\}_{i=1}^r$ is a permutation of $\{\alpha_i\}_{i=1}^r$. By rearrangement inequality,
\begin{align*}
\trace{\Sigma_A^{-1}Q^T\Sigma^2Q\Sigma_A^{-1}}=\sum_{i=1}^r\frac{\sigma_i^2}{\alpha_{\pi_i}^2}\geq\sum_{i=1}^r \frac{\sigma_i^2}{\alpha_i^2}.
\end{align*}
As a consequence, we have
\begin{align*}
&\|A\|_*\hspace{-0.02in} + \hspace{-0.02in}\frac{1}{2}\|X\|_F^2 \hspace{-0.02in}\geq\hspace{-0.02in} \sum_{i=1}^r \left(\alpha_i+\frac{\sigma_i^2}{2\alpha_i^2}\right)\hspace{-0.02in}=\hspace{-0.02in} \sum_{i=1}^r \left(\frac{1}{2}\alpha_i\hspace{-0.02in}+\hspace{-0.02in}\frac{1}{2}\alpha_i\hspace{-0.02in}+\hspace{-0.02in}\frac{\sigma_i^2}{2\alpha_i^2}\right)\\
&\geq{}\sum_{i=1}^r \frac{3}{2}\sigma_i^{\frac{2}{3}}=\frac{3}{2}\trace{\Sigma^{\frac{2}{3}}}.
\end{align*}
Regarding the general case of $\rank{A}\geq{}\rank{L}$, we can construct $A_1 = UU^TA$. By $AX=L$, $A_1X=L$. Since $\rank{A_1} = \rank{L}$, we have
\begin{align*}
&\|A\|_* + \frac{1}{2}\|X\|_F^2 \geq{}\|A_1\|_* + \frac{1}{2}\|X\|_F^2\\
&\geq{}\|A_1\|_* + \frac{1}{2}\|A_1^+L\|_F^2\geq\frac{3}{2}\trace{\Sigma^{\frac{2}{3}}}.
\end{align*}
Finally, the optimal value of $\frac{3}{2}\trace{\Sigma^{\frac{2}{3}}}$ is attained by $A_*=U\Sigma^{\frac{2}{3}}H^T$ and $X_*=H\Sigma^{\frac{1}{3}}V^T$, $\forall{}H^TH=\Id$.
\end{proof}

The next lemma will be used multiple times in the proofs presented in this paper.
\begin{lemm}\label{lem:basic:inverse}
Let $\Omega\subseteq\{1,\cdots,m\}\times\{1,\cdots,n\}$ and $\mathcal{P}$ be an orthogonal projection onto some subspace of $\Re^{m\times{}n}$. Then the following are equivalent:
\begin{itemize}
\item[1.] $\mathcal{P}\mathcal{P}_{\Omega}\mathcal{P}$ is invertible.
\item[2.] $\|\mathcal{P}\mathcal{P}_{\Omega}^\bot\mathcal{P}\|<1$.
\item[3.] $\mathcal{P}\cap{}\mathcal{P}_{\Omega}^\bot=\{0\}$.
\end{itemize}
\end{lemm}
\begin{proof} \textbf{1$\rightarrow$2:} Let $\mathrm{vec}(\cdot)$ denote the vectorization of a matrix formed by stacking the columns of the matrix into a single column vector. Suppose that the basis matrix associated with $\mathcal{P}$ is given by $P\in\Re^{mn\times{}r}, P^TP=\Id$; namely,
\begin{align*}
\mathrm{vec}(\mathcal{P}(M)) = PP^T\mathrm{vec}(M),\forall{}M\in\Re^{m\times{}n}.
\end{align*}
Denote $\delta_{ij}$ as in~\eqref{eq:delta} and define a diagonal matrix $D$ as
\begin{align*}
D = \mathrm{diag}(\delta_{11},\delta_{21},\cdots,\delta_{ij},\cdots,\delta_{mn})\in\Re^{mn\times{}mn}.
\end{align*}
Notice that
\begin{align*}
&\mathcal{P}(M) = \mathcal{P}(\sum_{i,j}\langle{}M,e_ie_j^T\rangle{}e_ie_j^T)=\sum_{i,j}\langle{}M,e_ie_j^T\rangle{}\mathcal{P}(e_ie_j^T),
\end{align*}
where $e_i$ is the $i$th standard basis and $\langle\cdot\rangle$ denotes the inner product between two matrices. With this notation, it is easy to see that
\begin{align*}
&[\mathrm{vec}(\mathcal{P}(e_1e_1^T)),\mathrm{vec}(\mathcal{P}(e_2e_1^T)),\cdots,\mathrm{vec}(\mathcal{P}(e_me_n^T))]=PP^T.
\end{align*}
Similarly, we have
\begin{align*}
\mathcal{P}\mathcal{P}_{\Omega}\mathcal{P}(M) = \sum_{i,j}\langle\mathcal{P}(M),e_ie_j^T\rangle(\delta_{ij}\mathcal{P}(e_ie_j^T)),
\end{align*}
and thereby
\begin{align*}
&\mathrm{vec}(\mathcal{P}\mathcal{P}_{\Omega}\mathcal{P}(M))= PP^TD\mathrm{vec}(\mathcal{P}(M))\\
&=PP^TDPP^T\mathrm{vec}(M).
\end{align*}
For $\mathcal{P}\mathcal{P}_{\Omega}\mathcal{P}$ to be invertible, the matrix $P^TDP$ must be positive definite. Because, whenever $P^TDP$ is singular, there exists $z\in\Re^{mn}$ that satisfies $z\neq0$ and $P^TDPz=0$, and thus there exists $M\in\mathcal{P}$ and $M\neq0$ such that $\mathcal{P}\mathcal{P}_{\Omega}\mathcal{P}(M)=0$; this contradicts the assumption that $\mathcal{P}\mathcal{P}_{\Omega}\mathcal{P}$ is invertible. Denote the minimal singular value of $P^TDP$ as $0<\sigma_{min}\leq1$. Since $P^TDP$ is positive definite, we have
\begin{align*}
&\|\mathcal{P}\mathcal{P}_{\Omega}^\bot\mathcal{P}(M)\|_F = \|\mathrm{vec}(\mathcal{P}\mathcal{P}_{\Omega}^\bot\mathcal{P}(M))\|_2\\
&= \|(\Id-P^TDP)P^T\mathrm{vec}(M)\|_2\leq(1-\sigma_{min})\|P^T\mathrm{vec}(M)\|_2 \\
&= (1-\sigma_{min})\|\mathcal{P}(M)\|_F,
\end{align*}
which gives that $\|\mathcal{P}\mathcal{P}_{\Omega}^\bot\mathcal{P}\|\leq1-\sigma_{min}<1$.

\textbf{2$\rightarrow$3:} Suppose that $M\in{}\mathcal{P}\cap{}\mathcal{P}_{\Omega}^\bot$, i.e., $M =\mathcal{P}(M)= \mathcal{P}_{\Omega}^\bot(M)$. Then we have $M=\mathcal{P}\mathcal{P}_{\Omega}^\bot\mathcal{P}(M)$ and thus
\begin{align*}
&\|M\|_F = \|\mathcal{P}\mathcal{P}_{\Omega}^\bot\mathcal{P}(M)\|_F\leq\|\mathcal{P}\mathcal{P}_{\Omega}^\bot\mathcal{P}\|\|M\|_F\leq\|M\|_F.
\end{align*}
Since $\|\mathcal{P}\mathcal{P}_{\Omega}^\bot\mathcal{P}\|<1$, the last equality above can hold only when $M=0$.

\textbf{3$\rightarrow$1:} Consider a nonzero matrix $M\in\mathcal{P}$. Then we have
\begin{align*}
&\|M\|_F^2 = \|\mathcal{P}(M)\|_F^2 = \|\mathcal{P}_{\Omega}\mathcal{P}(M)+\mathcal{P}_{\Omega}^\bot\mathcal{P}(M)\|_F^2\\
&=\|\mathcal{P}_{\Omega}\mathcal{P}(M)\|_F^2+\|\mathcal{P}_{\Omega}^\bot\mathcal{P}(M)\|_F^2,
\end{align*}
which gives that
\begin{align*}
&\|\mathcal{P}\mathcal{P}_{\Omega}^\bot\mathcal{P}(M)\|_F^2\leq\|\mathcal{P}_{\Omega}^\bot\mathcal{P}(M)\|_F^2=\|M\|_F^2 - \|\mathcal{P}_{\Omega}\mathcal{P}(M)\|_F^2.
\end{align*}
By $\mathcal{P}\cap{}\mathcal{P}_{\Omega}^\bot=\{0\}$, $\mathcal{P}_{\Omega}\mathcal{P}(M)\neq0$. Thus,
\begin{align*}
&\|\mathcal{P}\mathcal{P}_{\Omega}^\bot\mathcal{P}\|^2 \leq 1 - \inf_{\|M\|_F=1}\|\mathcal{P}_{\Omega}\mathcal{P}(M)\|_F^2<1.
\end{align*}
Provided that $\|\mathcal{P}\mathcal{P}_{\Omega}^\bot\mathcal{P}\|<1$, $\mathcal{I}+\sum_{i=1}^{\infty}(\mathcal{P}\mathcal{P}_{\Omega}^\bot\mathcal{P})^i$ is well defined. Notice that, for any $M\in\mathcal{P}$, the following holds:
\begin{align*}
&\mathcal{P}\mathcal{P}_{\Omega}\mathcal{P}(\mathcal{I}+\sum_{i=1}^{\infty}(\mathcal{P}\mathcal{P}_{\Omega}^\bot\mathcal{P})^i)(M)\\
&=\mathcal{P}(\mathcal{I}-\mathcal{P}\mathcal{P}_{\Omega}^\bot\mathcal{P})(\mathcal{I}+\sum_{i=1}^{\infty}(\mathcal{P}\mathcal{P}_{\Omega}^\bot\mathcal{P})^i)(M)\\
&=\mathcal{P}(\mathcal{I}+\sum_{i=1}^{\infty}(\mathcal{P}\mathcal{P}_{\Omega}^\bot\mathcal{P})^i-\mathcal{P}\mathcal{P}_{\Omega}^\bot\mathcal{P}-\sum_{i=2}^{\infty}(\mathcal{P}\mathcal{P}_{\Omega}^\bot\mathcal{P})^i)(M)\\
&=\mathcal{P}(M) = M.
\end{align*}
Similarly, it can be also proven that $(\mathcal{I}+\sum_{i=1}^{\infty}(\mathcal{P}\mathcal{P}_{\Omega}^\bot\mathcal{P})^i)$ $\mathcal{P}\mathcal{P}_{\Omega}\mathcal{P}(M)=M$. Hence, $\mathcal{I}+\sum_{i=1}^{\infty}(\mathcal{P}\mathcal{P}_{\Omega}^\bot\mathcal{P})^i$ is indeed the inverse operator of $\mathcal{P}\mathcal{P}_{\Omega}\mathcal{P}$.
\end{proof}

The lemma below is adapted from the arguments in~\cite{siam:stewart:1969}.
\begin{lemm}\label{lem:basic:pinv}
Let $A\in\mathbb{R}^{m\times{}p}$ be a matrix with column space $U$, and let $A_1 = A +\Delta$. If $\Delta\in{}U$ and $\|\Delta\|<1/\|A^+\|$ then
\begin{align*}
\rank{A_1} = \rank{A} \textrm{ and } \|A_1^+\|\leq{}\frac{\|A^+\|}{1 - \|A^+\|\|\Delta\|}.
\end{align*}
\end{lemm}
\begin{proof}By $\Delta\in{}U$,
\begin{align*}
A_1 = A + UU^T\Delta= A + AA^+\Delta = A(\Id + A^+\Delta).
\end{align*}
By $\|\Delta\|<1/\|A^+\|$, $\Id + A^+\Delta$ is invertible and thus $\rank{A_1} = \rank{A}$.

To prove the second claim, we denote by $V_1$ the row space of $A_1$. Then we have
\begin{align*}
V_1V_1^T = A_1^+A_1 = A_1^+A(\Id + A^+\Delta),
\end{align*}
which gives that $A_1^+A = V_1V_1^T(\Id + A^+\Delta)^{-1}$. Since $A_1\in{}U$, we have
\begin{align*}
A_1^+ = A_1^+UU^T =  A_1^+AA^+ = V_1V_1^T(\Id + A^+\Delta)^{-1}A^+,
\end{align*}
from which the conclusion follows.
\end{proof}
\subsection{Critical Lemmas}
The following lemma has a critical role in the proofs.
\begin{lemm}\label{lem:critical:inverse}
Let $L_0\in\Re^{m\times{}n}$ and $\Omega\subseteq\{1,\cdots,m\}\times\{1,\cdots,n\}$. Let the SVD of $L_0$ be $U_0\Sigma_0V_0^T$. Denote $\mathcal{P}_{U_0}(\cdot)=U_0U_0^T(\cdot)$ and $\mathcal{P}_{V_0}(\cdot)=(\cdot)V_0V_0^T$. Then we have the following:
\begin{itemize}
\item[1.] $\mathcal{P}_{U_0}\mathcal{P}_{\Omega}\mathcal{P}_{U_0}$ is invertible iff $U_0$ is $\Omega$-isomeric.
\item[2.] $\mathcal{P}_{V_0}\mathcal{P}_{\Omega}\mathcal{P}_{V_0}$ is invertible iff $V_0$ is $\Omega^T$-isomeric.
\end{itemize}
\end{lemm}
\begin{proof}
The above two claims are proven in the same way, and thereby we only present the proof of the first one. Since the operator $\mathcal{P}_{U_0}\mathcal{P}_{\Omega}\mathcal{P}_{U_0}$ is linear and $\mathcal{P}_{U_0}$ is a linear space of finite dimension, the sufficiency can be proven by showing that $\mathcal{P}_{U_0}\mathcal{P}_{\Omega}\mathcal{P}_{U_0}$ is an injection. That is, we need to prove that the following linear system has no nonzero solution:
\begin{align*}
\mathcal{P}_{U_0}\mathcal{P}_{\Omega}\mathcal{P}_{U_0}(M) = 0, \textrm{ s.t. }M\in\mathcal{P}_{U_0}.
\end{align*}
Assume that $\mathcal{P}_{U_0}\mathcal{P}_{\Omega}\mathcal{P}_{U_0}(M) = 0$. Then we have
\begin{align*}
U_0^T\mathcal{P}_{\Omega}(U_0U_0^TM) = 0.
\end{align*}
Denote the $i$th row and $j$th column of $U_0$ and $U_0^TM$ as $u_i^T$ and $b_j$, respectively; that is, $U_0 = [u_1^T;u_2^T;\cdots;u_m^T]$ and $U_0^TM = [b_1,b_2,\cdots,b_n]$. Define $\delta_{ij}$ as in~\eqref{eq:delta}. Then the $j$th column of $U_0^T\mathcal{P}_{\Omega}(U_0U_0^TM)$ is given by $(\sum_{i=1}^{m}\delta_{ij}u_iu_i^T)b_j$. By Lemma~\ref{lem:basic:positive}, the matrix $\sum_{i=1}^{m}\delta_{ij}u_iu_i^T$ is invertible. Hence, $U_0^T\mathcal{P}_{\Omega}(U_0U_0^TM) = 0$ implies that
\begin{align*}
b_j=0,\forall{}j=1,\cdots,n,
\end{align*}
i.e., $U_0^TM=0$. By the assumption of $M\in\mathcal{P}_{U_0}$, $M=0$.

It remains to prove the necessity. Assume $U_0$ is not $\Omega$-isomeric. By Lemma~\ref{lem:basic:positive}, there exists $j_1$ such that the matrix $\sum_{i=1}^{m}\delta_{ij_1}u_iu_i^T$ is singular and therefore has a nonzero null space. So, there exists $M_1\neq{}0$ such that $U_0^T\mathcal{P}_{\Omega}(U_0M_1)=0$. Let $M=U_0M_1$. Then we have $M\neq0$, $M\in\mathcal{P}_{U_0}$ and
\begin{align*}
\mathcal{P}_{U_0}\mathcal{P}_{\Omega}\mathcal{P}_{U_0}(M) = 0.
\end{align*}
This contradicts the assumption that $\mathcal{P}_{U_0}\mathcal{P}_{\Omega}\mathcal{P}_{U_0}$ is invertible. As a consequence, $U_0$ must be $\Omega$-isomeric.
\end{proof}

The next four lemmas establish some connections between the relative condition number and the operator norm.
\begin{lemm}\label{lem:critical:rnc2optnorm}
Let $L_0\in\Re^{m\times{}n}$ and $\Omega\subseteq\{1,\cdots,m\}\times\{1,\cdots,n\}$, and let the SVD of $L_0$ be $U_0\Sigma_0V_0^T$. Denote $\mathcal{P}_{U_0}(\cdot)=U_0U_0^T(\cdot)$ and $\mathcal{P}_{V_0}(\cdot)=(\cdot)V_0V_0^T$. If $L_0$ is $\Omega/\Omega^T$-isomeric then
\begin{align*}
&\|\mathcal{P}_{U_0}\mathcal{P}_{\Omega}^\bot\mathcal{P}_{U_0}\| = 1 - \gamma_{\Omega}(L_0),\textrm{ }\|\mathcal{P}_{V_0}\mathcal{P}_{\Omega}^\bot\mathcal{P}_{V_0}\| = 1 - \gamma_{\Omega^T}(L_0^T).
\end{align*}
\end{lemm}
\begin{proof} We only need to prove the first claim. Denote $\delta_{ij}$ as in~\eqref{eq:delta} and define a set of diagonal matrices $\{D_j\}_{j=1}^n$ as $D_j = \mathrm{diag}(\delta_{1j},\delta_{2j},\cdots,\delta_{mj})\in\Re^{m\times{}m}$. Denote the $j$th column of $\mathcal{P}_{U_0}(M)$ as $b_j$. Then we have
\begin{align*}
&\|[\mathcal{P}_{U_0}\mathcal{P}_{\Omega}^\bot\mathcal{P}_{U_0}(M)]_{:,j}\|_2 = \|U_0U_0^Tb_j - U_0(U_0^TD_jU_0)U_0^Tb_j\|_2\\
&=\|(\Id-U_0^TD_jU_0)U_0^Tb_j\|_2\leq\|(\Id-U_0^TD_jU_0)\|\|U_0^Tb_j\|_2.
\end{align*}
By Lemma~\ref{lem:basic:positive}, $U_0^TD_jU_0$ is positive definite. As a consequence, $\sigma_j\Id\preccurlyeq{}U_0^TD_jU_0\preccurlyeq\Id$,
where $\sigma_j>0$ is the minimal eigenvalue of $U_0^TD_jU_0$. By Lemma~\ref{lem:basic:rcn} and Definition~\ref{def:rcn:2}, $\sigma_j\geq{}\gamma_{\Omega}(L_0)$, $\forall{}1\leq{}j\leq{}n$. Thus,
\begin{align*}
&\|\mathcal{P}_{U_0}\mathcal{P}_{\Omega}^\bot\mathcal{P}_{U_0}(M)\|_F^2\leq\sum_{j=1}^{n}(1-\sigma_j)^2\|b_j\|_2^2\\
&\leq{}(1-\gamma_{\Omega}(L_0))^2\|\mathcal{P}_{U_0}(M)\|_F^2,
\end{align*}
where gives that $\|\mathcal{P}_{U_0}\mathcal{P}_{\Omega}^\bot\mathcal{P}_{U_0}\|\leq1-\gamma_{\Omega}(L_0)$.

It remains to prove that the value of $1-\gamma_{\Omega}(L_0)$ is attainable. Without loss of generality, assume that $j_1 = \arg\min_j\sigma_j$, i.e., $\sigma_{j_1} = \gamma_{\Omega}(L_0)$. Construct a $r_0\times{}r_0$ matrix $B$ with the $j_1$th column being the eigenvector corresponding to the smallest eigenvalue of $U_0^TD_{j_1}U_0$ and everywhere else being zero. Let $M_1 = U_0B$. Then it can be verified that $\|\mathcal{P}_{U_0}\mathcal{P}_{\Omega}^\bot\mathcal{P}_{U_0}(M_1)\|_F = (1-\gamma_{\Omega}(L_0))\|M_1\|_F$.
\end{proof}

\begin{lemm}\label{lem:critical:optnorm:big}
Let $L_0\in\Re^{m\times{}n}$ and $\Omega\subseteq\{1,\cdots,m\}\times\{1,\cdots,n\}$. Let the SVD of $L_0$ be $U_0\Sigma_0V_0^T$. Denote $\mathcal{P}_{U_0}(\cdot)=U_0U_0^T(\cdot)$ and $\mathcal{P}_{V_0}(\cdot)=(\cdot)V_0V_0^T$. If $L_0$ is $\Omega/\Omega^T$-isomeric then:
\begin{align*}
&\|(\mathcal{P}_{U_0}\mathcal{P}_{\Omega}\mathcal{P}_{U_0})^{-1}\mathcal{P}_{U_0}\mathcal{P}_{\Omega}\mathcal{P}_{U_0}^\bot\| =\sqrt{\frac{1}{\gamma_{\Omega}(L_0)}-1},\\
&\|(\mathcal{P}_{V_0}\mathcal{P}_{\Omega}\mathcal{P}_{V_0})^{-1}\mathcal{P}_{V_0}\mathcal{P}_{\Omega}\mathcal{P}_{V_0}^\bot\|=\sqrt{\frac{1}{\gamma_{\Omega^T}(L_0^T)} - 1}.
\end{align*}
\end{lemm}
\begin{proof} We shall prove the first claim. Let $M\in\Re^{m\times{}n}$. Denote the $j$th column of $M$ and $(\mathcal{P}_{U_0}\mathcal{P}_{\Omega}\mathcal{P}_{U_0})^{-1}\mathcal{P}_{U_0}\mathcal{P}_{\Omega}\mathcal{P}_{U_0}^\bot(M)$ as $b_j$ and $y_j$, respectively. Denote $\delta_{ij}$ as in~\eqref{eq:delta} and define a set of diagonal matrices $\{D_j\}_{j=1}^n$ as $D_j = \mathrm{diag}(\delta_{1j},\delta_{2j},\cdots,\delta_{mj})\in\Re^{m\times{}m}$. Then we have
\begin{align*}
&y_j = [(\mathcal{P}_{U_0}\mathcal{P}_{\Omega}\mathcal{P}_{U_0})^{-1}\mathcal{P}_{U_0}\mathcal{P}_{\Omega}\mathcal{P}_{U_0}^\bot(M)]_{:,j}\\
&= U_0(U_0^TD_jU_0)^{-1}U_0^TD_j(\Id - U_0U_0^T)b_j.
\end{align*}
It can be calculated that
\begin{align*}
&\|y_j\|_2^2 \leq \|(U_0^TD_jU_0)^{-1}U_0^TD_j(\Id - U_0U_0^T)\|^2\|b_j\|_2^2=\\
&\|(U_0^TD_jU_0)^{-1}U_0^TD_j(\Id - U_0U_0^T)DU_0(U_0^TD_jU_0)^{-1}\|\|b_j\|_2^2\\
&=\|(U_0^TD_jU_0)^{-1} - \Id\|\|b_j\|_2^2\leq\left(\frac{1}{\gamma_{\Omega}(L_0)}-1\right)\|b_j\|_2^2,
\end{align*}
which gives that
\begin{align*}
\|(\mathcal{P}_{U_0}\mathcal{P}_{\Omega}\mathcal{P}_{U_0})^{-1}\mathcal{P}_{U_0}\mathcal{P}_{\Omega}\mathcal{P}_{U_0}^\bot\|\leq\sqrt{\frac{1}{\gamma_{\Omega}(L_0)}-1}.
\end{align*}
Using a similar argument as in the proof of Lemma~\ref{lem:critical:rnc2optnorm}, it can be proven that the value of $\sqrt{1/\gamma_{\Omega}(L_0)-1}$ is attainable. To be more precise, assume without loss of generality that $j_1 = \arg\min_{j}\sigma_{j}$, where $\sigma_j$ is the smallest singular value of $U_0^TD_{j}U_0$. Denote by $\sigma^*$ and $v^*$ the largest singular value and the corresponding right singular vector of $(U_0^TD_jU_0)^{-1}U_0^TD_j(\Id - U_0U_0^T)$, respectively. Then the above justifications have already proven that $\sigma^*=\sqrt{1/\gamma_{\Omega}(L_0)-1}$. Construct an $m\times{}n$ matrix $M$ with the $j_1$th column being $v^*$ and everywhere else being zero. Then it can be verified that $\|(\mathcal{P}_{U_0}\mathcal{P}_{\Omega}\mathcal{P}_{U_0})^{-1}\mathcal{P}_{U_0}\mathcal{P}_{\Omega}\mathcal{P}_{U_0}^\bot(M)\|_F = \sqrt{1/\gamma_{\Omega}(L_0)-1}\|M\|_F$.
\end{proof}

\begin{lemm}\label{lem:critical:optnorm:ptpo}
Let $L_0\in\Re^{m\times{}n}$ and $\Omega\subseteq\{1,\cdots,m\}\times\{1,\cdots,n\}$, and let the SVD of $L_0$ be $U_0\Sigma_0V_0^T$. Denote $\mathcal{P}_{T_0}(\cdot)=U_0U_0^T(\cdot)+(\cdot)V_0V_0^T-U_0U_0^T(\cdot)V_0V_0^T$. If $L_0$ is $\Omega/\Omega^T$-isomeric then
\begin{align*}
\|\mathcal{P}_{T_0}\mathcal{P}_{\Omega}^\bot\mathcal{P}_{T_0}\| \leq{}2(1-\gamma_{\Omega,\Omega^T}(L_0)).
\end{align*}
\end{lemm}
\begin{proof} Using the same arguments as in the proof of Lemma~\ref{lem:basic:inverse}, it can be proven that $\|\mathcal{P}\mathcal{P}_{\Omega}^\bot\mathcal{P}\| = \|\mathcal{P}\mathcal{P}_{\Omega}^\bot\|^2$, with $\mathcal{P}$ being any orthogonal projection onto a subspace of $\mathbb{R}^{m\times{}n}$. Thus, we have the following
 \begin{align*}
& \|\mathcal{P}_{T_0}\mathcal{P}_{\Omega}^\bot\mathcal{P}_{T_0}\| = \|\mathcal{P}_{T_0}\mathcal{P}_{\Omega}^\bot\|^2 = \sup_{\|M\|_F=1}\|\mathcal{P}_{T_0}\mathcal{P}_{\Omega}^\bot(M)\|_F^2 \\ &= \sup_{\|M\|_F=1}\|\mathcal{P}_{U_0}\mathcal{P}_{\Omega}^\bot(M) + \mathcal{P}_{U_0}^\bot\mathcal{P}_{V_0}\mathcal{P}_{\Omega}^\bot(M)\|_F^2\\
&= \sup_{\|M\|_F=1}(\|\mathcal{P}_{U_0}\mathcal{P}_{\Omega}^\bot(M)\|_F^2 + \|\mathcal{P}_{U_0}^\bot\mathcal{P}_{V_0}\mathcal{P}_{\Omega}^\bot(M)\|_F^2)\\
&\leq\sup_{\|M\|_F=1}\|\mathcal{P}_{U_0}\mathcal{P}_{\Omega}^\bot(M)\|_F^2 + \sup_{\|M\|_F=1}\|\mathcal{P}_{V_0}\mathcal{P}_{\Omega}^\bot(M)\|_F^2 \\
&= \|\mathcal{P}_{U_0}\mathcal{P}_{\Omega}^\bot\|^2 + \|\mathcal{P}_{V_0}\mathcal{P}_{\Omega}^\bot\|^2,
 \end{align*}
which, together with Lemma~\ref{lem:critical:rnc2optnorm}, gives that
 \begin{align*}
&\|\mathcal{P}_{T_0}\mathcal{P}_{\Omega}^\bot\mathcal{P}_{T_0}\| \leq \|\mathcal{P}_{U_0}\mathcal{P}_{\Omega}^\bot\mathcal{P}_{U_0}\| + \|\mathcal{P}_{V_0}\mathcal{P}_{\Omega}^\bot\mathcal{P}_{V_0}\|\\
&= 1 - \gamma_{\Omega} (L_0) + 1 - \gamma_{\Omega^T} (L_0^T)\leq2(1 - \gamma_{\Omega,\Omega^T}(L_0))
\end{align*}
\end{proof}

\begin{lemm}\label{lem:critical:optnorm:invpt}
Let $L_0\in\Re^{m\times{}n}$ and $\Omega\subseteq\{1,\cdots,m\}\times\{1,\cdots,n\}$, and let the SVD of $L_0$ be $U_0\Sigma_0V_0^T$. Denote $\mathcal{P}_{T_0}(\cdot)=U_0U_0^T(\cdot)+(\cdot)V_0V_0^T-U_0U_0^T(\cdot)V_0V_0^T$. If the operator $\mathcal{P}_{T_0}\mathcal{P}_{\Omega}\mathcal{P}_{T_0}$ is invertible, then we have
\begin{align*}
\|(\mathcal{P}_{T_0}\mathcal{P}_{\Omega}\mathcal{P}_{T_0})^{-1}\mathcal{P}_{T_0}\mathcal{P}_{\Omega}\mathcal{P}_{T_0}^\bot\| = \sqrt{\frac{1}{1-\|\mathcal{P}_{T_0}\mathcal{P}_{\Omega}^\bot\mathcal{P}_{T_0}\|}-1}.
\end{align*}
\end{lemm}
\begin{proof} We shall use again the two notations, $\mathrm{vec}(\cdot)$ and $D$, defined in the proof of Lemma~\ref{lem:basic:inverse}. Let $P\in\mathbb{R}^{mn\times{}r}$ be a column-wisely orthonormal matrix such that $\mathrm{vec}(\mathcal{P}_{T_0}(M)) = PP^T\mathrm{vec}(M)$, $\forall{}M$. Since $\mathcal{P}_{T_0}\mathcal{P}_{\Omega}\mathcal{P}_{T_0}$ is invertible, it follows that $P^TDP$ is positive definite. Denote by $\sigma_{min}(\cdot)$ the smallest singular value of a matrix. Then we have the following:
\begin{align*}
&\|(\mathcal{P}_{T_0}\mathcal{P}_{\Omega}\mathcal{P}_{T_0})^{-1}\mathcal{P}_{T_0}\mathcal{P}_{\Omega}\mathcal{P}_{T_0}^\bot\|^2 \\
&= \|P(P^TDP)^{-1}P^TD(\Id-PP^T)\|^2\\
&=\|P(P^TDP)^{-1}P^TD(\Id-PP^T)DP(P^TDP)^{-1}P^T\| \\
&= \|(P^TDP)^{-1}-\Id\| = \frac{1}{\sigma_{min}(P^TDP)} -1 \\
&= \frac{1}{1 - \|P^T(\Id - D)P\|} - 1=\frac{1}{1 - \|\mathcal{P}_{T_0}\mathcal{P}_{\Omega}^\bot\mathcal{P}_{T_0}\|} - 1.
\end{align*}
\end{proof}

The following lemma is more general than Theorem~\ref{thm:fnorm}.
\begin{lemm}\label{lem:critical:uinorm}
Let $L_0\in\Re^{m\times{}n}$ and $\Omega\subseteq\{1,\cdots,m\}\times\{1,\cdots,n\}$. Consider the following convex problem:
\begin{align}\label{eq:uinorm}
\min_{X} \norm{X}_{UI},\textrm{ s.t. }\mathcal{P}_{\Omega}(AX-L_0)=0,
\end{align}
where $\norm{\cdot}_{UI}$ generally denotes a convex unitary invariant norm and $A\in\Re^{m\times{}p}$ is given. If $L_0\in\mathrm{span}\{A\}$ and $A$ is $\Omega$-isomeric then $X_0=A^+L_0$ is the unique minimizer to the convex optimization problem in~\eqref{eq:uinorm}.
\end{lemm}
\begin{proof}
Denote the SVD of $A$ as $U_A\Sigma_AV_A^T$. Then it follows from $\mathcal{P}_{\Omega}(AX-L_0)=0$ and $L_0\in\mathrm{span}\{A\}$ that
\begin{align*}
\mathcal{P}_{U_A}\mathcal{P}_{\Omega}\mathcal{P}_{U_A}(AX-L_0) = 0.
\end{align*}
By Lemma~\ref{lem:basic:L02U} and Lemma~\ref{lem:critical:inverse}, $\mathcal{P}_{U_A}\mathcal{P}_{\Omega}\mathcal{P}_{U_A}$ is invertible and thus $AX = L_0$. Hence, $\mathcal{P}_{\Omega}(AX-L_0)=0$ is equivalent to $AX=L_0$. Notice, that Theorem 4.1 of~\cite{tpami_2013_lrr} actually holds for any convex unitary invariant norms. That is,
\begin{align*}
A^+L_0 = \arg\min_{X} \|X\|_{UI}, \textrm{ s.t. } AX = L_0,
\end{align*}
which implies that $A^+L_0$ is the unique minimizer to the problem in~\eqref{eq:uinorm}.
\end{proof}

\subsection{Proofs of Theorems~\ref{thm:iso},~\ref{thm:iso:necessary} and~\ref{thm:rcn:bound}}
We need to use some notations as follows. Let the SVD of $L_0$ be $U_0\Sigma_0V_0^T$. Denote $\mathcal{P}_{U_0}(\cdot)=U_0U_0^T(\cdot)$, $\mathcal{P}_{V_0}(\cdot)=(\cdot)V_0V_0^T$ and $\mathcal{P}_{T_0}(\cdot) = \mathcal{P}_{U_0}(\cdot)+\mathcal{P}_{V_0}(\cdot)-\mathcal{P}_{U_0}\mathcal{P}_{V_0}(\cdot)$.

\begin{proof}({\bf proof of Theorem~\ref{thm:iso}})
Define an operator $\mathcal{H}$ in the same way as in~\citep{Candes:2009:math}:
\begin{align*}
\mathcal{H} = \mathcal{P}_{T_0} - \frac{1}{\rho_0}\mathcal{P}_{T_0}\mathcal{P}_{\Omega_{\mathcal{A}}}\mathcal{P}_{T_0}.
\end{align*}
According to Theorem 4.1 of~\citep{Candes:2009:math}, there exists some numerical constant $c>0$ such that the inequality,
\begin{align*}
\|\mathcal{H}\|\leq\sqrt{\frac{c\mu_0r_0\log{n_1}}{\rho_0n_2}},
\end{align*}
holds with probability at least $1-n_1^{-10}$ provided that the right hand side is smaller than 1. So, $\|\mathcal{H}\|<1$ provided that
\begin{align*}
\rho_0>\frac{c\mu_0r_0\log{n_1}}{n_2}.
\end{align*}
When $\|\mathcal{H}\|<1$, we have
\begin{align*}
&\|\mathcal{P}_{T_0}\mathcal{P}_{\Omega}^\bot\mathcal{P}_{T_0}\| = \|\rho_0\mathcal{H}+(1-\rho_0)\mathcal{P}_{T_0}\|\\
&\leq{}\rho_0\|\mathcal{H}\|+(1-\rho_0)\|\mathcal{P}_{T_0}\|<1.
\end{align*}
Since $\mathcal{P}_{U_0}(\cdot)=\mathcal{P}_{U_0}\mathcal{P}_{T_0}(\cdot)=\mathcal{P}_{T_0}\mathcal{P}_{U_0}(\cdot)$, we have
\begin{align*}
&\|\mathcal{P}_{U_0}\mathcal{P}_{\Omega}^\bot\mathcal{P}_{U_0}\| = \|\mathcal{P}_{U_0}\mathcal{P}_{T_0}\mathcal{P}_{\Omega}^\bot\mathcal{P}_{T_0}\mathcal{P}_{U_0}\|\leq\|\mathcal{P}_{T_0}\mathcal{P}_{\Omega}^\bot\mathcal{P}_{T_0}\|<1.
\end{align*}
Due to the virtues of Lemma~\ref{lem:basic:inverse}, Lemma~\ref{lem:critical:inverse} and Lemma~\ref{lem:basic:L02U}, it can be concluded that $L_0$ is $\Omega$-isometric with probability at least $1-n_1^{-10}$. In a similar way, it can be also proven that $L_0^T$ is $\Omega^T$-isometric with probability at least $1-n_1^{-10}$.
\end{proof}
\begin{proof}({\bf proof of Theorem~\ref{thm:iso:necessary}}) When $L_0$ is not $\Omega$-isomeric, Lemma~\ref{lem:basic:L02U} and Lemma~\ref{lem:critical:inverse} give that $\mathcal{P}_{U_0}\mathcal{P}_{\Omega}\mathcal{P}_{U_0}$ is not invertible. By Lemma~\ref{lem:basic:inverse}, $\mathcal{P}_{U_0}\cap{}\mathcal{P}_\Omega^\bot\neq\{0\}$. Thus, there exists $\Delta\neq0$ that satisfies $\Delta\in\mathcal{P}_{U_0}$ and $\Delta\in\mathcal{P}_\Omega^\bot$. Now construct $L = L_0 + \Delta$. Then we have $L\neq{}L_0$, $\mathcal{P}_\Omega(L) = \mathcal{P}_\Omega(L_0)$ and $\rank{L} = \rank{\mathcal{P}_{U_0}(L_0+\Delta)}\leq\rank{L_0}$. Since $\mathcal{P}_{U_0}\cap{}\mathcal{P}_\Omega^\bot$ is a nonempty linear space, there are indeed infinitely many choices for $L$.
\end{proof}

\begin{proof}({\bf proof of Theorem~\ref{thm:rcn:bound}}) Using the same arguments as in the proof of Theorem~\ref{thm:iso}, we conclude that the following holds with probability at least $1-n_1^{-10}$:
\begin{align*}
\|\mathcal{P}_{U_0}\mathcal{P}_{\Omega}^\bot\mathcal{P}_{U_0}\|<1-\rho_0+\frac{\rho_0}{\sqrt{\alpha}},
\end{align*}
which, together with Lemma~\ref{lem:critical:rnc2optnorm}, gives that $\gamma_{\Omega}(L_0)>(1-1/\sqrt{\alpha})\rho_0$. Similarity, it can be also proven that $\gamma_{\Omega^T}(L_0^T)>(1-1/\sqrt{\alpha})\rho_0$ with probability at least $1-n_1^{-10}$.
\end{proof}

\subsection{Proof of Theorem~\ref{thm:iso:rcn}}
Let the SVD of $L_0$ be $U_0\Sigma_0V_0^T$. Denote the $i$th row of $U_0$ as $u_i^T$, i.e., $U_0 = [u_1^T;u_2^T;\cdots;u_{m}^T]$. Define $\delta_{ij}$ as in~\eqref{eq:delta}, and define a collection of diagonal matrices $\{D_j\}_{j=1}^{n}$ as $D_j = \mathrm{diag}(\delta_{1j},\delta_{2j},\cdots,\delta_{mj})\in\mathbb{R}^{m\times{}m}$. With these notations, we shall show that the operator norm of $\mathcal{P}_{U_0}\mathcal{P}_{\Omega}^\bot\mathcal{P}_{U_0}$ can be bounded from above. Considering the $j$th column of $\mathcal{P}_{U_0}\mathcal{P}_{\Omega}^\bot\mathcal{P}_{U_0}(X), \forall{}X, j$, we have
\begin{align*}
&[\mathcal{P}_{U_0}\mathcal{P}_{\Omega}^\bot\mathcal{P}_{U_0}(X)]_{:,j} = U_0U_0^T(\Id-D_j)U_0U_0^T[X]_{:,j},
\end{align*}
which gives that
\begin{align*}
\|[\mathcal{P}_{U_0}\mathcal{P}_{\Omega}^\bot\mathcal{P}_{U_0}(X)]_{:,j}\|_2\leq\|U_0U_0^T(\Id - D_j)U_0U_0^T\|\|[X]_{:,j}\|_2.
\end{align*}
Since the diagonal of $D_j$ has at most $(1-\rho)m$ zeros,
\begin{align*}
&\|U_0U_0^T(\Id - D_j)U_0U_0^T\|= \|\sum_{i=1}^{m_1}(1-\delta_{ij})u_iu_i^T\|\\
&\leq\sum_{i=1}^{m}(1-\delta_{ij})\|u_iu_i^T\|\leq(1-\rho)\mu_0r_0,
\end{align*}
where the last inequality follows from the definition of coherence. Thus, we have
\begin{align*}
\|\mathcal{P}_{U_0}\mathcal{P}_{\Omega}^\bot\mathcal{P}_{U_0}\|\leq(1-\rho)\mu_0r_0.
\end{align*}
Similarly, based on the assumption that at least $\rho{}n$ entries in each row of $L_0$ are observed, we have
\begin{align*}
\|\mathcal{P}_{V_0}\mathcal{P}_{\Omega}^\bot\mathcal{P}_{V_0}\|\leq(1-\rho)\mu_0r_0.
\end{align*}
By the assumption $\rho>1 - (1-\alpha)/(\mu_0r_0)$,
\begin{align*}
\|\mathcal{P}_{U_0}\mathcal{P}_{\Omega}^\bot\mathcal{P}_{U_0}\|<1-\alpha\quad\textrm{and}\quad\|\mathcal{P}_{V_0}\mathcal{P}_{\Omega}^\bot\mathcal{P}_{V_0}\|<1-\alpha.
\end{align*}
By Lemma~\ref{lem:basic:inverse} and Lemma~\ref{lem:critical:inverse}, $L_0$ is $\Omega/\Omega^T$-isomeric. In addition, it follows from Lemma~\ref{lem:critical:rnc2optnorm} that $\gamma_{\Omega,\Omega^T}(L_0)>\alpha$.
\subsection{Proofs of Theorems~\ref{thm:l2} and~\ref{thm:fnorm}}
Theorem~\ref{thm:fnorm} is indeed an immediate corollary of Lemma~\ref{lem:critical:uinorm}. So we only prove Theorem~\ref{thm:l2}.
\begin{proof}By $y_0\in\mathcal{S}_0\subseteq\mathrm{span}\{A\}$, $y_0=AA^+y_0$ and therefore $y_b = A_bA^+y_0$.
That is, $x_0=A^+y_0$ is a feasible solution to the problem in~\eqref{eq:l2}. Provided that $y_b\in\Re^k$ and the dictionary matrix $A$ is $k$-isomeric, Definition~\ref{def:iso:k} gives that $\rank{A_b} = \rank{A}$, which implies that
\begin{align*}
\mathrm{span}\{A_b^T\}=\mathrm{span}\{A^T\}.
\end{align*}
On the other hand, it is easy to see that $A^+y_0\in\mathrm{span}\{A^T\}$. Hence, there exists a dual vector $w\in\Re^p$ that obeys
\begin{align*}
A_b^Tw = A^+y_0, \textrm{ i.e., } A_b^Tw \in\partial\frac{1}{2}\|A^+y_0\|_2^2.
\end{align*}
By standard convexity arguments~\cite{book:convex}, $x_0=A^{+}y_0$ is an optimal solution to the problem in~\eqref{eq:l2}. Since the squared $\ell_2$ norm is a strongly convex function, it follows that the optimal solution to~\eqref{eq:l2} is unique.
\end{proof}
\subsection{Proof of Theorem~\ref{thm:convex}}
\begin{proof} Let the SVD of $L_0$ be $U_0\Sigma_0V_0^T$. Denote $\mathcal{P}_{T_0}(\cdot)=U_0U_0^T(\cdot)+(\cdot)V_0V_0^T-U_0U_0^T(\cdot)V_0V_0^T$. Since $\gamma_{\Omega,\Omega^T}(L_0)>0.5$, it follows from Lemma~\ref{lem:critical:optnorm:ptpo} that $\|\mathcal{P}_{T_0}\mathcal{P}_{\Omega}^\bot\mathcal{P}_{T_0}\|$ is strictly smaller than 1. By Lemma~\ref{lem:basic:inverse}, $\mathcal{P}_{T_0}\mathcal{P}_{\Omega}\mathcal{P}_{T_0}$ is invertible and $T_0\cap{}\Omega^\bot = \{0\}$. Given $\gamma_{\Omega,\Omega^T}(L_0)>0.75$, Lemma~\ref{lem:critical:optnorm:invpt} and  Lemma~\ref{lem:critical:optnorm:ptpo} imply that
\begin{align*}
&\|(\mathcal{P}_{T_0}\mathcal{P}_{\Omega}\mathcal{P}_{T_0})^{-1}\mathcal{P}_{T_0}\mathcal{P}_{\Omega}\mathcal{P}_{T_0}^\bot\| = \sqrt{\frac{1}{1-\|\mathcal{P}_{T_0}\mathcal{P}_{\Omega}^\bot\mathcal{P}_{T_0}\|}-1}\\
&\leq\sqrt{\frac{1}{2\gamma_{\Omega,\Omega^T}(L_0)-1}-1}<1.
\end{align*}
Next, we shall consider a feasible solution $L=L_0+\Delta$ and show that the objective strictly increases unless $\Delta=0$. By $\mathcal{P}_{\Omega}(\Delta) = 0$, $\mathcal{P}_{\Omega}\mathcal{P}_{T_0}(\Delta) = -\mathcal{P}_{\Omega}\mathcal{P}_{T_0}^\bot(\Delta)$. Since the operator $\mathcal{P}_{T_0}\mathcal{P}_{\Omega}\mathcal{P}_{T_0}$ is invertible, we have
\begin{align*}
\mathcal{P}_{T_0}(\Delta) = -(\mathcal{P}_{T_0}\mathcal{P}_{\Omega}\mathcal{P}_{T_0})^{-1}\mathcal{P}_{T_0}\mathcal{P}_{\Omega}\mathcal{P}_{T_0}^\bot(\Delta).
\end{align*}
By $\|(\mathcal{P}_{T_0}\mathcal{P}_{\Omega}\mathcal{P}_{T_0})^{-1}\mathcal{P}_{T_0}\mathcal{P}_{\Omega}\mathcal{P}_{T_0}^\bot\|<1$, $\|\mathcal{P}_{T_0}(\Delta)\|_*<\|\mathcal{P}_{T_0}^\bot(\Delta)\|_*$ holds unless $\mathcal{P}_{T_0}^\bot(\Delta)=0$. By the convexity of the nuclear norm,
\begin{align*}
&\|L_0+\Delta\|_* - \|L_0\|_*\geq{}\langle{}\Delta,U_0V_0^T+W\rangle,
\end{align*}
where $W\in{}\mathcal{P}_{T_0}^\bot$ and $\|W\|\leq1$. Due to the duality between the nuclear norm and operator norm, we can construct a $W$ such that $\langle{}\Delta,W\rangle=\|\mathcal{P}_{T_0}^\bot(\Delta_0)\|_*$. Thus,
\begin{align*}
&\|L_0+\Delta\|_* - \|L_0\|_*\geq{}\|\mathcal{P}_{T_0}^\bot(\Delta)\|_* - \|\mathcal{P}_{U_0}\mathcal{P}_{V_0}(\Delta)\|_*\\
&\geq\|\mathcal{P}_{T_0}^\bot(\Delta)\|_* - \|\mathcal{P}_{T_0}(\Delta)\|_*.
\end{align*}
Hence, $\|L_0+\Delta\|_*$ is strictly greater than $\|L_0\|_*$ unless $\Delta\in{}T_0$. Since $T_0\cap\Omega^\bot=\{0\}$, it follows that $L_0$ is the unique minimizer to the problem in~\eqref{eq:numin}.
\end{proof}
\subsection{Proof of Theorem~\ref{thm:isodp:f}}
\begin{proof}
Since $A_0 = U_0\Sigma_0^{\frac{1}{2}}Q^T$ and  $X_0= Q\Sigma_0^{\frac{1}{2}}V_0^T$, we have the following: 1) $A_0X_0=L_0$; 2) $L_0\in\mathrm{span}\{A_0\}$ and $A_0$ is $\Omega$-isomeric; 3) $L_0^T\in\mathrm{span}\{X_0^T\}$ and $X_0^T$ is $\Omega^T$-isomeric. Hence, according to Lemma~\ref{lem:critical:uinorm}, we have
\begin{align*}
&X_0 = A_0^+L_0=\arg\min_{X} \|X\|_F^2,\textrm{ s.t. }\mathcal{P}_{\Omega}(A_0X - L_0)=0,\\
&A_0 = L_0X_0^+=\arg\min_{A} \|A\|_F^2,\textrm{ s.t. }\mathcal{P}_{\Omega}(AX_0 - L_0)=0.
\end{align*}
Hence, $(A_0,X_0)$ is a critical point to the problem in~\eqref{eq:isodp:f}.

It remains to prove the second claim. Suppose that $(A=A_0+\Delta_0, X = X_0+E_0)$ with $\|\Delta_0\|\leq\varepsilon$ and $\|E_0\|\leq\varepsilon$ is a feasible solution to~\eqref{eq:isodp:f}. We want to prove that
\begin{align*}
\frac{1}{2}(\|A\|_F^2+\|X\|_F^2) \geq \frac{1}{2}(\|A_0\|_F^2+\|X_0\|_F^2)
\end{align*}
holds for some small $\varepsilon$, and show that the equality can hold only if $AX=L_0$. Denote
\begin{align}\label{eq:temp:notation}
&\mathcal{P}_{U_0}(\cdot)=U_0U_0^T(\cdot), \mathcal{P}_{V_0}(\cdot)=(\cdot)V_0V_0^T,\\\nonumber
&\mathcal{P}_1 = (\mathcal{P}_{U_0}\mathcal{P}_{\Omega}\mathcal{P}_{U_0})^{-1}\mathcal{P}_{U_0}\mathcal{P}_{\Omega}\mathcal{P}_{U_0}^\bot,\\\nonumber
&\mathcal{P}_2 = (\mathcal{P}_{V_0}\mathcal{P}_{\Omega}\mathcal{P}_{V_0})^{-1}\mathcal{P}_{V_0}\mathcal{P}_{\Omega}\mathcal{P}_{V_0}^\bot.
\end{align}
Define
\begin{align}\label{eq:temp:notation:1}
&\bar{A}_0 = A_0 + \mathcal{P}_{U_0}(\Delta_0) \textrm{ and } \bar{X}_0 = X_0 + \mathcal{P}_{V_0}(E_0).
\end{align}
Provided that $\varepsilon<\min(1/\|A_0^+\|, 1/\|X_0^+\|)$, it follows from Lemma~\ref{lem:basic:pinv} that
\begin{align}\label{eq:temp:notation:pseinv}
&\rank{\bar{A}_0} = \rank{\bar{X}_0} = r_0,\\\nonumber
&\|\bar{A}_0^+\|\leq\frac{\|A_0^+\|}{1-\|A_0^+\|\varepsilon}\textrm{ and }\|\bar{X}_0^+\|\leq\frac{\|X_0^+\|}{1-\|X_0^+\|\varepsilon}.
\end{align}
By $\mathcal{P}_{\Omega}(AX-L_0)=0$,
\begin{align*}
\mathcal{P}_{\Omega}(A_0E_0+\Delta_0X_0+\Delta_0E_0)=0.
\end{align*}
Then it can be manipulated that
\begin{align*}
&\mathcal{P}_{\Omega}(\bar{A}_0E_0) \\
&= -\mathcal{P}_{\Omega}(\Delta_0\bar{X}_0- \mathcal{P}_{U_0}\mathcal{P}_{V_0}(\Delta_0E_0) + \mathcal{P}_{U_0}^\bot\mathcal{P}_{V_0}^\bot(\Delta_0E_0)).
\end{align*}
Since $\mathcal{P}_{U_0}\mathcal{P}_{\Omega}\mathcal{P}_{U_0}$ is invertible, we have
\begin{align}\label{eq:temp:p1}
&\mathcal{P}_{V_0}^\bot(\bar{A}_0E_0) = -\mathcal{P}_{V_0}^\bot(\mathcal{P}_{U_0}\mathcal{P}_{\Omega}\mathcal{P}_{U_0})^{-1}\mathcal{P}_{U_0}\mathcal{P}_{\Omega}(\Delta_0\bar{X}_0\\\nonumber
&-\mathcal{P}_{U_0}\mathcal{P}_{V_0}(\Delta_0E_0) + \mathcal{P}_{U_0}^\bot\mathcal{P}_{V_0}^\bot(\Delta_0E_0)) \\\nonumber
&= -\mathcal{P}_{V_0}^\bot\mathcal{P}_1\mathcal{P}_{U_0}^\bot(\Delta_0\bar{X}_0) - \mathcal{P}_{V_0}^\bot\mathcal{P}_1\mathcal{P}_{U_0}^\bot\mathcal{P}_{V_0}^\bot(\Delta_0E_0)
\end{align}
Similarly, by the invertibility of $\mathcal{P}_{V_0}\mathcal{P}_{\Omega}\mathcal{P}_{V_0}$,
\begin{align}\label{eq:temp:p2}
&\mathcal{P}_{U_0}^\bot(\Delta_0\bar{X}_0)\\\nonumber
&= -\mathcal{P}_{U_0}^\bot\mathcal{P}_2\mathcal{P}_{V_0}^\bot(\bar{A}_0E_0) - \mathcal{P}_{U_0}^\bot\mathcal{P}_2\mathcal{P}_{U_0}^\bot\mathcal{P}_{V_0}^\bot(\Delta_0E_0).
\end{align}
The combination of~\eqref{eq:temp:p1} and~\eqref{eq:temp:p2} gives that
\begin{align*}
&\mathcal{P}_{V_0}^\bot(\bar{A}_0E_0) = \mathcal{P}_{V_0}^\bot\mathcal{P}_1\mathcal{P}_2\mathcal{P}_{V_0}^\bot(\bar{A}_0E_0) + \\\nonumber &\mathcal{P}_{V_0}^\bot(\mathcal{P}_1\mathcal{P}_2-\mathcal{P}_1)\mathcal{P}_{U_0}^\bot\mathcal{P}_{V_0}^\bot(\Delta_0E_0).
\end{align*}
By $\rank{\bar{A}_0}=r_0=p$,
\begin{align*}
\mathcal{P}_{U_0}^\bot\mathcal{P}_{V_0}^\bot(\Delta_0E_0) = \mathcal{P}_{U_0}^\bot\mathcal{P}_{V_0}^\bot(\Delta_0\bar{A}_0^+\bar{A}_0E_0).
\end{align*}
By Lemma~\ref{lem:critical:optnorm:big} and the assumption of $\gamma_{\Omega,\Omega^T}(L_0)>0.5$, $\|\mathcal{P}_1\|<1$ and $\|\mathcal{P}_2\|<1$. Thus,
\begin{align*}
&\|\mathcal{P}_{V_0}^\bot(\bar{A}_0E_0)\|\leq\|\mathcal{P}_1\mathcal{P}_2\|\|\mathcal{P}_{V_0}^\bot(\bar{A}_0E_0)\|\\
&+\varepsilon(\|\mathcal{P}_1\mathcal{P}_2\|+\|\mathcal{P}_1\|)\|\bar{A}_0^+\|\|\mathcal{P}_{V_0}^\bot(\bar{A}_0E_0)\|\\
&\leq\left(\frac{1}{\gamma_{\Omega,\Omega^T}(L_0)}-1+\frac{2\varepsilon\|A_0^+\|}{1-\|A_0^+\|\varepsilon}\right)\|\mathcal{P}_{V_0}^\bot(\bar{A}_0E_0)\|.
\end{align*}
Let
\begin{align*}
\varepsilon < \min\left(\frac{1}{2\|A_0^+\|}, \frac{2\gamma_{\Omega,\Omega^T}(L_0)-1}{4\|A_0^+\|\gamma_{\Omega,\Omega^T}(L_0)}\right).
\end{align*}
Then we have that $\|\mathcal{P}_{V_0}^\bot(\bar{A}_0E_0)\|<\|\mathcal{P}_{V_0}^\bot(\bar{A}_0E_0)\|$ strictly holds unless $\mathcal{P}_{V_0}^\bot(\bar{A}_0E_0)=0$. Since $\rank{\bar{A}_0}=r_0=p$, $\mathcal{P}_{V_0}^\bot(\bar{A}_0E_0)=0$ simply leads to $E_0\in\mathcal{P}_{V_0}$. Hence,
\begin{align*}
A_0E_0+\Delta_0X_0+\Delta_0E_0 \in\mathcal{P}_{V_0}\cap\mathcal{P}_{\Omega}^\bot = \{0\},
\end{align*}
which implies that $AX = L_0$. Thus, we finally have
\begin{align*}
\frac{1}{2}(\|A\|_F^2+\|X\|_F^2)\geq\|L_0\|_*=\frac{1}{2}(\|A_0\|_F^2+\|X_0\|_F^2),
\end{align*}
where the inequality follows from $\|AX\|_*=\min_{A,X}\frac{1}{2}(\|A\|_F^2+\|X\|_F^2)$~\cite{siam_2010_minirank}.
\end{proof}
\subsection{Proof of Theorem~\ref{thm:isodp}}
\begin{proof}
Since $A_0 = U_0\Sigma_0^{\frac{2}{3}}Q^T$ and  $X_0= Q\Sigma_0^{\frac{1}{3}}V_0^T$, we have the following: 1) $A_0X_0=L_0$; 2) $L_0\in\mathrm{span}\{A_0\}$ and $A_0$ is $\Omega$-isomeric; 3) $L_0^T\in\mathrm{span}\{X_0^T\}$ and $X_0^T$ is $\Omega^T$-isomeric. Due to Lemma~\ref{lem:critical:uinorm}, we have
\begin{align*}
&X_0 = A_0^+L_0=\arg\min_{X} \|X\|_F^2,\textrm{ s.t. }\mathcal{P}_{\Omega}(A_0X - L_0)=0,\\
&A_0 = L_0X_0^+=\arg\min_{A} \|A\|_*,\textrm{ s.t. }\mathcal{P}_{\Omega}(AX_0 - L_0)=0.
\end{align*}
Hence, $(A_0,X_0)$ is a critical point to the problem in~\eqref{eq:isodp}.

Regarding the second claim, we consider a feasible solution $(A=A_0+\Delta_0, X = X_0+E_0)$, with $\|\Delta_0\|\leq\varepsilon$ and $\|E_0\|\leq\varepsilon$. Define $\mathcal{P}_{U_0}$, $\mathcal{P}_{V_0}$, $\mathcal{P}_1$, $\mathcal{P}_2$, $\bar{A}_0$ and $\bar{X}_0$ in the same way as in~\eqref{eq:temp:notation} and~\eqref{eq:temp:notation:1}. Note that the statements in~\eqref{eq:temp:notation:pseinv} still hold in the general case of $p\geq{}r_0$. Denote the SVD of $\bar{X}_0$ as $\bar{Q}\bar{\Sigma}\bar{V}_0^T$. Then we have $V_0V_0^T = \bar{V}_0\bar{V}_0^T$. Denote
\begin{align*}
P_{\bar{Q}} = \bar{Q}\bar{Q}^T \textrm{ and }P_{\bar{Q}}^\bot = \Id - \bar{Q}\bar{Q}^T.
\end{align*}
Denote the condition number of $X_0$ as $\tau_0$. With these notations, we shall finish the proof by exploring two cases.
\subsubsection{Case 1: $\|\mathcal{P}_{U_0}^\bot(\Delta_0)P_{\bar{Q}}^\bot\|_*\geq2\tau_0\|\mathcal{P}_{U_0}^\bot(\Delta_0)P_{\bar{Q}}\|_*$}
Denote the SVD of $L_0\bar{X}_0^+$ as $\tilde{U}_0\tilde{\Sigma}\tilde{Q}^T$. Then we have
\begin{align*}
\tilde{U}_0\tilde{U}_0^T = U_0U_0^T \textrm{ and }\tilde{Q}\tilde{Q}^T = \bar{Q}\bar{Q}^T.
\end{align*}
By the convexity of the nuclear norm,
\begin{align}\label{eq:temp:ded:1}
&\|A\|_* - \|L_0\bar{X}_0^+\|_*=\|A_0+\Delta_0\|_* - \|L_0\bar{X}_0^+\|_*\\\nonumber
&\geq{}\langle{}A_0+\Delta_0-L_0\bar{X}_0^+,\tilde{U}_0\tilde{Q}^T+W\rangle,
\end{align}
where $W\in\mathbb{R}^{m\times{}p}$, $\tilde{U}_0^TW = 0$, $W\tilde{Q} = 0$ and $\|W\|\leq1$. Due to the duality between the nuclear norm and operator
norm, we can construct a $W$ such that
\begin{align}\label{eq:temp:ded:2}
\langle{}\Delta_0,W\rangle=\|\mathcal{P}_{U_0}^\bot(\Delta_0)P_{\bar{Q}}^\bot\|_*.
\end{align}
We also have
\begin{align*}
&\langle{}A_0\hspace{-0.02in}+\hspace{-0.02in}\Delta_0\hspace{-0.02in}-\hspace{-0.02in}L_0\bar{X}_0^+,\tilde{U}_0\tilde{Q}^T\rangle\hspace{-0.02in}=\hspace{-0.02in}\langle{}\Delta_0+A_0E_0\bar{X}_0^+,\tilde{U}_0\tilde{Q}^T\rangle\\
&=\langle{}\Delta_0\bar{X}_0\bar{X}_0^++A_0E_0\bar{X}_0^+,\tilde{U}_0\tilde{Q}^T\rangle,
\end{align*}
which gives that
\begin{align}\label{eq:temp:ded:3}
&\mathrm{abs}(\langle{}A_0+\Delta_0-L_0\bar{X}_0^+,\tilde{U}_0\tilde{Q}^T\rangle)\leq\|\bar{X}_0^+\|\|\mathcal{P}_{U_0}\mathcal{P}_{V_0}\\\nonumber
&(\Delta_0\bar{X}_0+A_0E_0)\|_*\leq\|\bar{X}_0^+\|\|\Delta_0\bar{X}_0+\mathcal{P}_{V_0}(A_0E_0)\|_*,
\end{align}
where we denote by $\mathrm{abs}(\cdot)$ the absolute value of a real number. By $\mathcal{P}_{\Omega}(A_0E_0+\Delta_0X_0+\Delta_0E_0)=0$,
\begin{align*}
&\Delta_0\bar{X}_0+\mathcal{P}_{V_0}(A_0E_0) = -\mathcal{P}_2\mathcal{P}_{V_0}^\bot(A_0E_0) - \mathcal{P}_2\mathcal{P}_{V_0}^\bot(\Delta_0E_0)\\
&=-\mathcal{P}_2(-\mathcal{P}_{V_0}^\bot\mathcal{P}_1(\Delta_0X_0+\Delta_0E_0) - \mathcal{P}_{V_0}^\bot\mathcal{P}_{U_0}(\Delta_0E_0))\\
&- \mathcal{P}_2\mathcal{P}_{V_0}^\bot(\Delta_0E_0)=\mathcal{P}_2\mathcal{P}_1(\Delta_0X_0+\Delta_0E_0) - \mathcal{P}_2\mathcal{P}_{U_0}^\bot(\Delta_0E_0)\\
&= \mathcal{P}_2\mathcal{P}_1(\Delta_0\bar{X}_0) + \mathcal{P}_2\mathcal{P}_1\mathcal{P}_{V_0}^\bot(\Delta_0E_0) - \mathcal{P}_2\mathcal{P}_{U_0}^\bot(\Delta_0E_0).
\end{align*}
By Lemma~\ref{lem:critical:optnorm:big} and the assumption of $\gamma_{\Omega,\Omega^T}(L_0)>0.5$, $\|\mathcal{P}_1\|<1$ and $\|\mathcal{P}_2\|<1$. As a result, we have
\begin{align}\label{eq:temp:ded:3b}
&\|\Delta_0\bar{X}_0+\mathcal{P}_{V_0}(A_0E_0)\|_* \\\nonumber
&\leq \|\mathcal{P}_{U_0}^\bot(\Delta_0\bar{X}_0)\|_*+2\|\mathcal{P}_{U_0}^\bot(\Delta_0E_0)\|_*.
\end{align}
Let
\begin{align*}
\varepsilon < \min\left(\frac{0.1\|X_0\|}{1+1.1\tau_0},\frac{0.175}{\|X_0^+\|}\right).
\end{align*}
Due to~\eqref{eq:temp:ded:3},~\eqref{eq:temp:ded:3b} and the assumption of $\|\mathcal{P}_{U_0}^\bot(\Delta_0)P_{\bar{Q}}^\bot\|_*\geq2\tau_0\|\mathcal{P}_{U_0}^\bot(\Delta_0)P_{\bar{Q}}\|_*$, it can be calculated that
\begin{align}\label{eq:temp:ded:4}
&\mathrm{abs}(\langle{}A_0+\Delta_0-L_0\bar{X}_0^+,\tilde{U}_0\tilde{Q}^T\rangle)\\\nonumber
&\leq\|\bar{X}_0^+\|\|\mathcal{P}_{U_0}^\bot(\Delta_0\bar{X}_0)\|_*+2\|\bar{X}_0^+\|\|\mathcal{P}_{U_0}^\bot(\Delta_0(P_{\bar{Q}}+P_{\bar{Q}}^\bot)E_0)\|_*\\\nonumber
&\leq\|\bar{X}_0^+\|\|\bar{X}_0\|\|\mathcal{P}_{U_0}^\bot(\Delta_0)P_{\bar{Q}}\|_*+2\varepsilon\|\bar{X}_0^+\|\|\mathcal{P}_{U_0}^\bot(\Delta_0)P_{\bar{Q}}\|_*\\\nonumber
&+2\varepsilon\|\bar{X}_0^+\|\|\mathcal{P}_{U_0}^\bot(\Delta_0)P_{\bar{Q}}^\bot\|_*\leq1.1\tau_0\|\mathcal{P}_{U_0}^\bot(\Delta_0)P_{\bar{Q}}\|_*\\\nonumber &+0.2\tau_0\|\mathcal{P}_{U_0}^\bot(\Delta_0)P_{\bar{Q}}\|_*+0.35\|\mathcal{P}_{U_0}^\bot(\Delta_0)P_{\bar{Q}}^\bot\|_*\\\nonumber
&\leq(0.65+0.35)\|\mathcal{P}_{U_0}^\bot(\Delta_0)P_{\bar{Q}}^\bot\|_*=\|\mathcal{P}_{U_0}^\bot(\Delta_0)P_{\bar{Q}}^\bot\|_*.
\end{align}
Now, combining~\eqref{eq:temp:ded:1},~\eqref{eq:temp:ded:2} and~\eqref{eq:temp:ded:4}, we have
\begin{align*}
&\|A\|_* - \|L_0\bar{X}_0^+\|_*\geq\|\mathcal{P}_{U_0}^\bot(\Delta_0)P_{\bar{Q}}^\bot\|_*\\
&-\mathrm{abs}(\langle{}A_0+\Delta_0-L_0\bar{X}_0^+,\tilde{U}_0\tilde{Q}^T\rangle)\geq0,
\end{align*}
which, together with Lemma~\ref{lem:basic:ax}, simply leads to
\begin{align*}
&\|A\|_* + \frac{1}{2}\|X\|_F^2 = (\|A\|_*-\|L_0\bar{X}_0^+\|_*)\\
&+(\|L_0\bar{X}_0^+\|_*+\frac{1}{2}\|X\|_F^2)\geq\|L_0\bar{X}_0^+\|_*+\frac{1}{2}\|\bar{X}_0\|_F^2\\
&\geq{}\frac{3}{2}\trace{\Sigma_0^{\frac{2}{3}}}=\|A_0\|_* + \frac{1}{2}\|X_0\|_F^2.
\end{align*}
For the equality of $\|A\|_* + 0.5\|X\|_F^2=\|A_0\|_* + 0.5\|X_0\|_F^2$ to hold, at least, $\|X\|_F = \|\bar{X}_0\|_F$ must be obeyed, which implies that $E_0\in\mathcal{P}_{V_0}$. Hence, we have $A_0E_0+\Delta_0X_0+\Delta_0E_0 \in\mathcal{P}_{V_0}\cap\mathcal{P}_{\Omega}^\bot = \{0\}$,
which gives that $AX = L_0$.
\subsubsection{Case 2: $\|\mathcal{P}_{U_0}^\bot(\Delta_0)P_{\bar{Q}}^\bot\|_*\leq2\tau_0\|\mathcal{P}_{U_0}^\bot(\Delta_0)P_{\bar{Q}}\|_*$}
Using a similar manipulation as in the proof of Theorem~\ref{thm:isodp:f}, we have
\begin{align*}
&\mathcal{P}_{U_0}^\bot(\Delta_0\bar{X}_0) = \mathcal{P}_{U_0}^\bot\mathcal{P}_2\mathcal{P}_1\mathcal{P}_{U_0}^\bot(\Delta_0\bar{X}_0)+ \\ &\mathcal{P}_{U_0}^\bot(\mathcal{P}_2\mathcal{P}_1-\mathcal{P}_2)\mathcal{P}_{U_0}^\bot\mathcal{P}_{V_0}^\bot(\Delta_0E_0)=\mathcal{P}_{U_0}^\bot\mathcal{P}_2\mathcal{P}_1\mathcal{P}_{U_0}^\bot(\Delta_0\bar{X}_0)\\
&+\mathcal{P}_{U_0}^\bot(\mathcal{P}_2\mathcal{P}_1-\mathcal{P}_2)\mathcal{P}_{U_0}^\bot\mathcal{P}_{V_0}^\bot(\Delta_0P_{\bar{Q}}E_0 + \Delta_0P_{\bar{Q}}^\bot{}E_0) .
\end{align*}
Due to Lemma~\ref{lem:critical:optnorm:big} and the assumption of $\gamma_{\Omega,\Omega^T}(L_0)>0.5$, we have $\|\mathcal{P}_1\|<1$ and $\|\mathcal{P}_2\|<1$. By the assumption of $\|\mathcal{P}_{U_0}^\bot(\Delta_0)P_{\bar{Q}}^\bot\|_*\leq2\tau_0\|\mathcal{P}_{U_0}^\bot(\Delta_0)P_{\bar{Q}}\|_*$,
\begin{align*}
&\|\mathcal{P}_{U_0}^\bot(\Delta_0\bar{X}_0)\|_*\leq\|\mathcal{P}_2\mathcal{P}_1\|\|\mathcal{P}_{U_0}^\bot(\Delta_0\bar{X}_0)\|_*\\
&+(4\tau_0+2)\varepsilon\|\mathcal{P}_{U_0}^\bot(\Delta_0)P_{\bar{Q}}\|_*= \|\mathcal{P}_2\mathcal{P}_1\|\|\mathcal{P}_{U_0}^\bot(\Delta_0\bar{X}_0)\|_*\\
&+(4\tau_0+2)\varepsilon\|\mathcal{P}_{U_0}^\bot(\Delta_0)\bar{X}_0\bar{X}_0^+\|_* \\
&\leq\left(\frac{1}{\gamma_{\Omega,\Omega^T}(L_0)}-1+\frac{(4\tau_0+2)\varepsilon\|X_0^+\|}{1-\|X_0^+\|\varepsilon}\right)\|\mathcal{P}_{U_0}^\bot(\Delta_0\bar{X}_0)\|_*.
\end{align*}
Let
\begin{align*}
\varepsilon < \min\left(\frac{1}{2\|X_0^+\|}, \frac{2\gamma_{\Omega,\Omega^T}(L_0)-1}{(8\tau_0+4)\|X_0^+\|\gamma_{\Omega,\Omega^T}(L_0)}\right).
\end{align*}
Then $\|\mathcal{P}_{U_0}^\bot(\Delta_0\bar{X}_0)\|_*<\|\mathcal{P}_{U_0}^\bot(\Delta_0\bar{X}_0)\|_*$ strictly holds unless $\mathcal{P}_{U_0}^\bot(\Delta_0\bar{X}_0)=0$. That is,
\begin{align*}
\mathcal{P}_{U_0}^\bot(\Delta_0)P_{\bar{Q}} = 0 \textrm{ and thus } \mathcal{P}_{U_0}^\bot(\Delta_0)P_{\bar{Q}}^\bot = 0.
\end{align*}
Hence, we have $\mathcal{P}_{U_0}^\bot(\Delta_0)=0$, which simply leads to
\begin{align*}
A_0E_0+\Delta_0X_0+\Delta_0E_0 \in\mathcal{P}_{U_0}\cap\mathcal{P}_{\Omega}^\bot = \{0\},
\end{align*}
and which gives that $AX = L_0$. By Lemma~\ref{lem:basic:ax},
\begin{align*}
&\|A\|_*+\frac{1}{2}\|X\|_F^2 \geq{}\frac{3}{2}\trace{\Sigma_0^{\frac{2}{3}}}=\|A_0\|_* + \frac{1}{2}\|X_0\|_F^2.
\end{align*}
\end{proof}
\begin{figure}
\begin{center}
\includegraphics[width=0.48\textwidth]{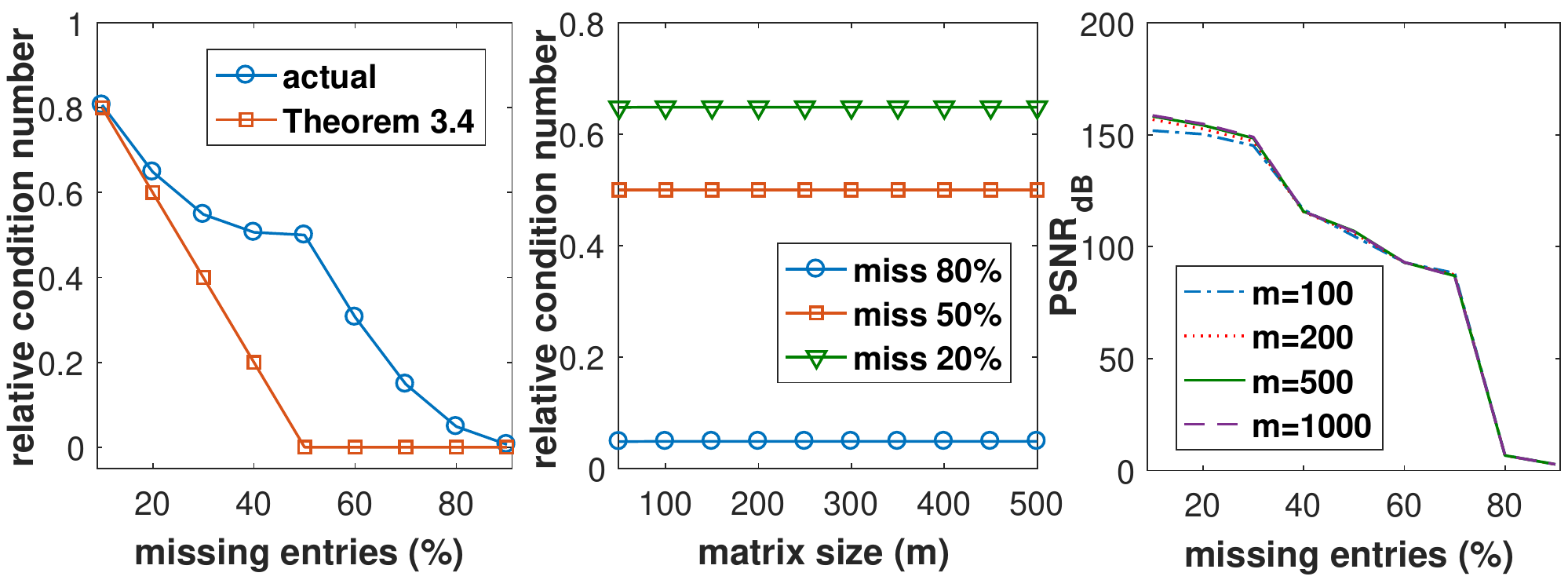}\vspace{-0.15in}
\caption{Left: The relative condition number $\gamma_{\Omega,\Omega^T}(L_0)$ vs the missing rate $1-\rho_0$ at $m=500$. Middle: The relative condition number vs the matrix size $m$. Right: Plotting the recovery performance of convex optimization as a function of the missing rate.}\label{fig:rcn}\vspace{-0.2in}
\end{center}
\end{figure}
\section{Experiments}\label{sec:exp}
\subsection{Investigating the Relative Condition Number}\label{sec:exp:rcn}
To study the properties of the relative condition number, we generate a vector $x\in\mathbb{R}^{m}$ according to the model $[x]_t = \sin(2t\pi/m)$, $t=1,\cdots,m$. That is, $x$ is a univariate time series of dimension $m$. We consider the forecasting tasks of recovering $x$ from a collection of $l$ observations, $\{[x]_t\}_{t=1}^{l}$, where $l=\rho_0m$ varies from $0.1m$ to $0.9m$ with step size $0.1m$. Let $y\in\mathbb{R}^{m}$ be the mask vector of the sampling operator, i.e., $[y]_t$ is 1 if $[x]_t$ is observed and 0 otherwise. In order to recover $x$, it suffices to recover its \emph{convolution matrix}~\cite{liu:tip:2014}. Thus, the forecasting tasks here can be converted to matrix completion problems, with
\begin{align*}
L_0 = \mathcal{A}(x)\quad\textrm{and}\quad\Omega=\mathrm{supp}(\mathcal{A}(y)),
\end{align*}
where $\mathcal{A}(\cdot)$ is the convolution matrix of a tensor\footnote{Unlike~\cite{liu:tip:2014}, we adopt here the circulant boundary condition. Thus, the $j$th column of $\mathcal{A}(x)$ is simply the vector obtained by circularly shifting the elements in $x$ by $j-1$ positions.}, and $\mathrm{supp}(\cdot)$ is the support set of a matrix. In this example, $L_0\in\mathbb{R}^{m\times{}m}$ is a circulant matrix that is perfectly incoherent and low rank; namely, $\rank{L_0}\equiv2$ and $\mu(L_0)\equiv1$, $\forall{}m>2$. Moreover, each column and each row of $\Omega$ have exactly a cardinality of $\rho_0m$. We use the convex program~\eqref{eq:numin} to restore $L_0$ from the given observations.

The results are shown in Figure~\ref{fig:rcn}. It can be seen that the relative condition number is independent of the matrix sizes and monotonously deceases as the missing rate grows. As we can see from the right hand side of Figure~\ref{fig:rcn}, the recovery performance visibly declines when the missing rate exceeds $30\%$ (i.e., $\rho_0<0.7$), which approximately corresponds to $\gamma_{\Omega,\Omega^T}(L_0)<0.55$. When $\rho_0<0.3$ (which corresponds approximately to $\gamma_{\Omega,\Omega^T}(L_0)<0.15$), matrix completion totally breaks down. These results illustrate that relative well-conditionedness is important for guaranteeing the success of matrix completion in practice. Of course, the lower bound on $\gamma_{\Omega,\Omega^T}(L_0)$ would depend on the characteristics of data, and the condition $\gamma_{\Omega,\Omega^T}(L_0)>0.75$ proven in Theorem~\ref{thm:convex} is just a universal bound for guaranteeing exact recovery in the worst case. In addition, the estimate given in Theorem~\ref{thm:iso:rcn} is accurate only when the missing rate is low, as shown in the left part of Figure~\ref{fig:rcn}.

Among the other things, it is worth noting that the sampling complexity does not decrease as the matrix size $m$ grows. This phenomenon is in conflict with the uniform sampling based matrix completion theories, which prove that a small fraction of $O((\log{m})^2/m)$ entries should suffice to recover $L_0$~\cite{Chen:2015:tit}, and which implies that the sampling complexity should decrease to zero when the matrix size $m$ goes to infinity. Hence, as aforementioned, the theories built upon uniform sampling are no longer applicable when applying to the deterministic missing data patterns.
\subsection{Results on Randomly Generated Matrices}
\begin{figure}
\begin{center}
\includegraphics[width=0.45\textwidth]{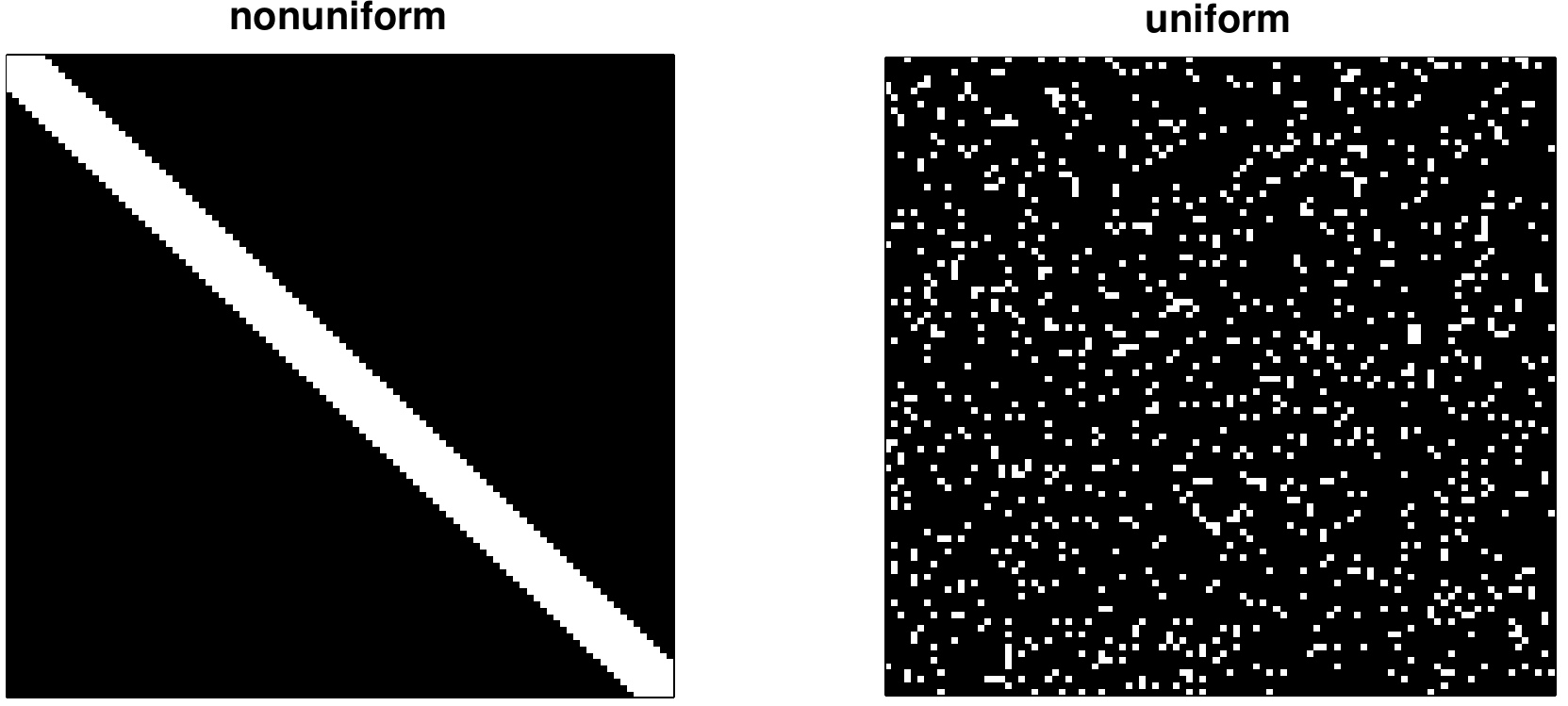}\vspace{-0.15in}
\caption{Visualizing the configurations of $\Omega$ used in our simulations. The white points correspond to the locations of the observed entries. In these two examples, 90\% entries of the matrix are missing.}\label{fig:location}\vspace{-0.2in}
\end{center}
\end{figure}
\begin{figure}
\begin{center}
\subfigure[nonuniform]{\includegraphics[width=0.48\textwidth]{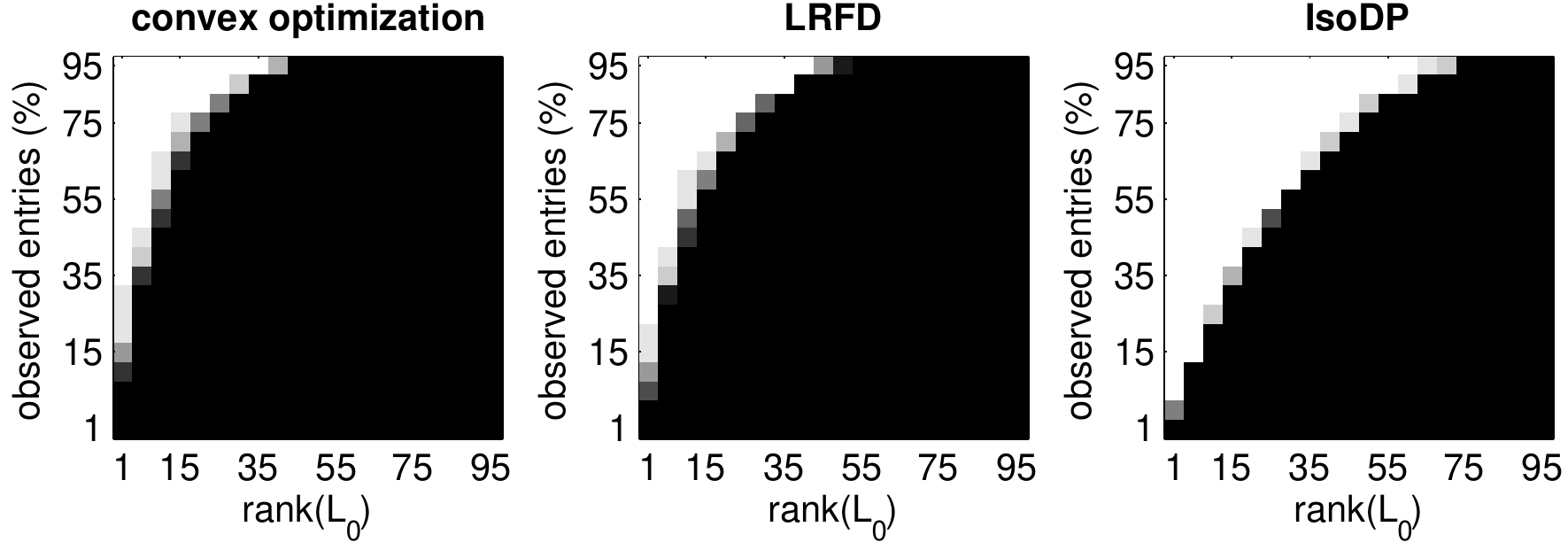}}
\subfigure[uniform]{\includegraphics[width=0.48\textwidth]{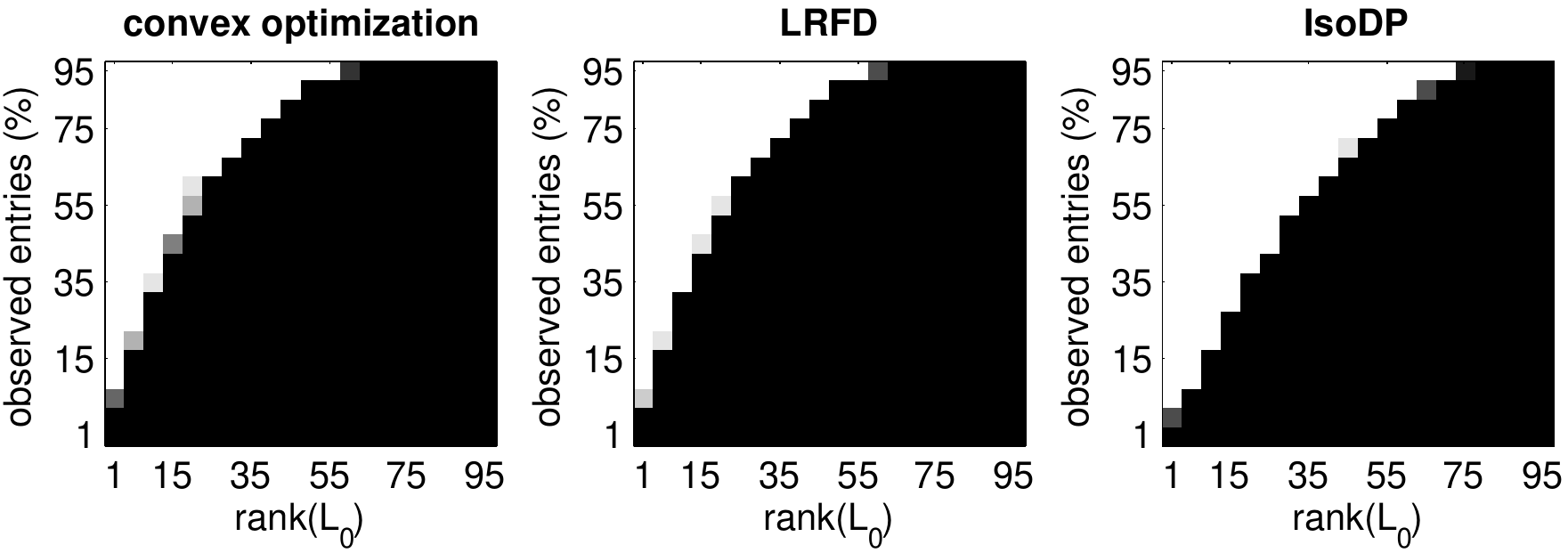}}\vspace{-0.15in}
\caption{Comparing IsoDP with convex optimization and LRFD. The numbers plotted on the above figures are the success rates within 20 random trials. The white and black areas mean ``succeed" and ``fail", respectively. Here, the success is in a sense that $\mathrm{PSNR_{dB}}$ $\geq$ 40.}\label{fig:cmp}\vspace{-0.2in}
\end{center}
\end{figure}
\begin{figure}
\begin{center}
\includegraphics[width=0.45\textwidth]{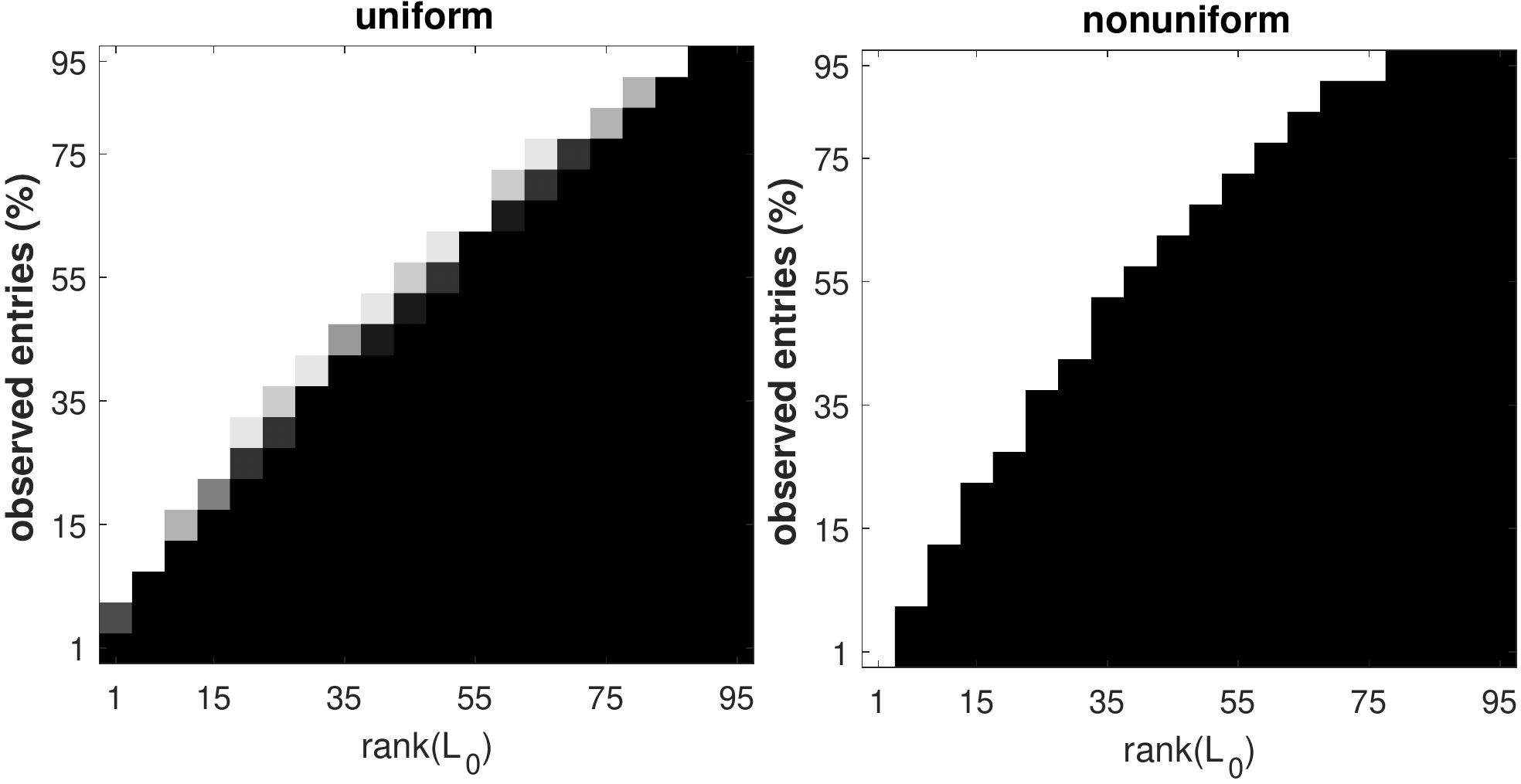}\vspace{-0.15in}
\caption{Visualizing the regions in which the isomeric condition holds.}\label{fig:iso}\vspace{-0.2in}
\end{center}
\end{figure}
To evaluate the performance of various matrix completion methods, we generate a collection of $m\times{}n$ ($m=n=100$) target matrices according to $L_0=BC$, where $B\in\Re^{m\times{}r_0}$ and $C\in\Re^{r_0\times{}n}$ are $\mathcal{N}(0,1)$ matrices. The rank of $L_0$, i.e., $r_0$, is configured as $r_0=1, 5, 10, \cdots, 90, 95$. Regarding the sampling set $\Omega$ consisting of the locations of the observed entries, we consider two settings: One is to create $\Omega$ by using a Bernoulli model to randomly sample a subset from $\{1,\cdots,m\}\times\{1,\cdots,n\}$ (referred to as ``uniform''), the other is to let the locations of the observed entries be centered around the main diagonal of a matrix (referred to as ``nonuniform''). Figure~\ref{fig:location} shows how the sampling set $\Omega$ looks like. The observation fraction is set as $|\Omega|/(mn)=0.01,0.05,\cdots,0.9, 0.95$. To show the advantages of IsoDP, we include for comparison two prevalent methods: convex optimization~\cite{Candes:2009:math} and Low-Rank Factor Decomposition (LRFD)~\cite{liu:tsp:2016}. The same as IsoDP, these two methods do not assume that rank of $L_0$ either. When $p = m$ and the identity matrix is used to initialize the dictionary $A$, the bilinear program~\eqref{eq:isodp:f} does not outperform convex optimization, thereby we exclude it from the comparison.

The accuracy of recovery, i.e., the similarity between $L_0$ and $\hat{L}_0$, is measured by Peak Signal-to-Noise Ratio ($\mathrm{PSNR_{dB}}$). Figure~\ref{fig:cmp} compares IsoDP to convex optimization and LRFD. It can be seen that IsoDP works distinctly better than the competing methods. Namely, while handling the nonuniformly missing data, the number of matrices successfully restored by IsoDP is 102\% and 71\% more than convex optimization and LRFD, respectively. While dealing with the missing entries chosen uniformly at random, in terms of the number of successfully restored matrices, IsoDP outperforms both convex optimization and LRFD by 44\%. These results verify the effectiveness of IsoDP. Figure~\ref{fig:iso} plots the regions where the isometric condition is valid. By comparing Figure~\ref{fig:cmp} to Figure~\ref{fig:iso}, it can be seen that the recovery performance of IsoDP has not reached the upper limit defined by isomerism. That is, there is still some room left for improvement.
\subsection{Results on Motion Data}
\begin{figure}
\begin{center}
\includegraphics[width=0.45\textwidth]{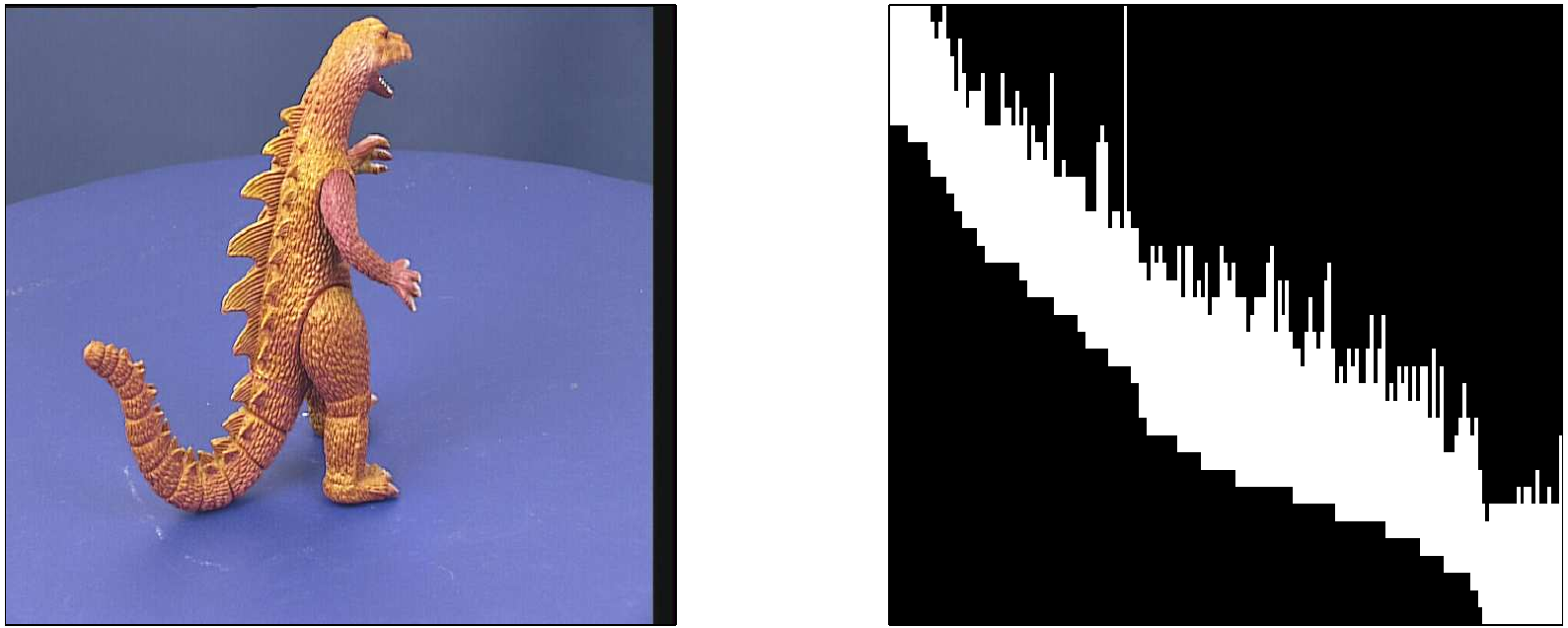}
\vspace{-0.15in}\caption{An example image from the Oxford dinosaur sequence and the locations of the observed entries in the data matrix of trajectories. In this dataset, 74.29\% entries of the trajectory matrix are missing.}\label{fig:din}\vspace{-0.2in}
\end{center}
\end{figure}
\begin{figure}
\begin{center}
\includegraphics[width=0.45\textwidth]{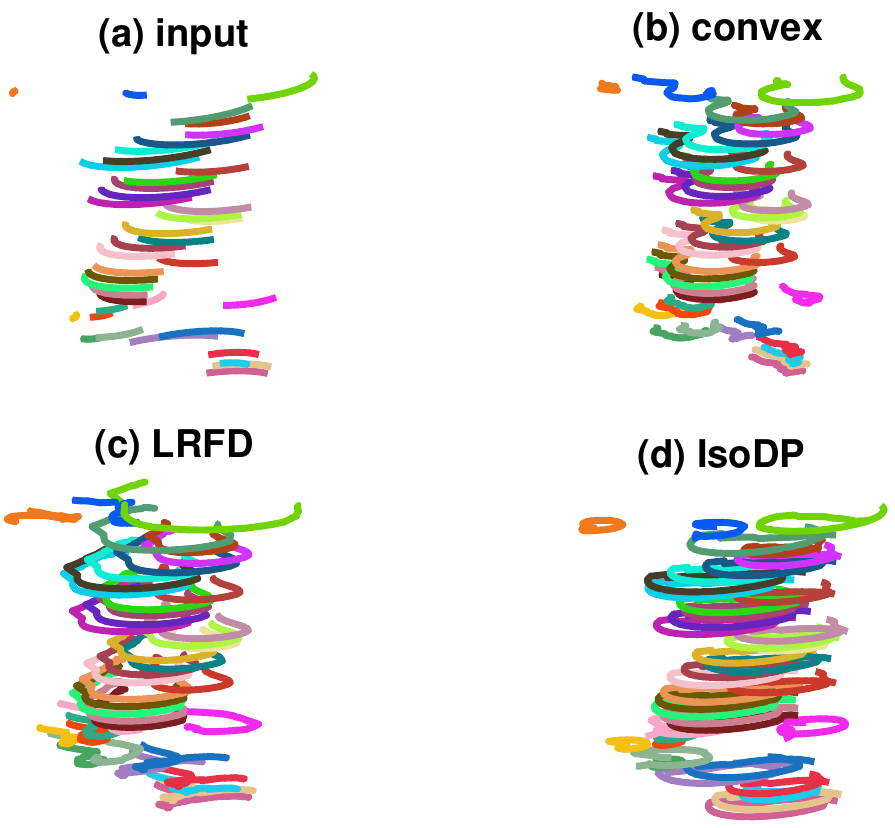}
\vspace{-0.15in}\caption{Some examples of the originally incomplete and fully restored trajectories. (a) The original incomplete trajectories. (b) The trajectories restored by convex optimization~\cite{Candes:2009:math}. (c) The trajectories restored by LRFD~\cite{liu:tsp:2016}. (d) The trajectories restored by IsoDP.}\label{fig:dinosaur:res}\vspace{-0.2in}
\end{center}
\end{figure}
We now consider the Oxford dinosaur sequence\footnote{Available at http://www.robots.ox.ac.uk/$\sim$vgg/data1.html}, which contains in total 72 image frames corresponding to 4983 track points observed by at least 2 among 36 views. The values of the observations range from 8.86 to 629.82. We select 195 track points which are observed by at least 6 views for experiment, resulting in a $72\times{}195$ trajectory matrix 74.29\% entries of which are missing (see Figure~\ref{fig:din}). The tracked dinosaur model is rotating around its center, and thus the true trajectories should form complete circles~\cite{zheng:cvpr:2012}.
\begin{table}
\caption{Mean square error (MSE) on the Oxford dinosaur sequence. Here, the rank of a matrix is estimated by $\#\{i |\sigma_i\geq10^{-4}\sigma_1\}$, where $\sigma_1\geq\sigma_2\geq\cdots$ are the singular values of the matrix. The regularization parameter in each method is manually tuned such that the rank of the restored matrix meets a certain value. Here, the MSE values are evaluated on the training data (i.e., observed entries).}\label{tb:motion}\vspace{-0.2in}
\begin{center}
\begin{tabular}{|c|c|c|c|}\hline
rank of the                       &                                   &                   &\\
restored matrix                   &convex optimization                &LRFD               &IsoDP\\\hline
 6                                &426.1369                           &28.4649            &\textbf{0.6140}\\
 7                                &217.9963                           &21.6968            &\textbf{0.4682}\\
 8                                &136.7643                           &17.2269            &\textbf{0.1480}\\
 9                                &94.4673                            &13.954             &\textbf{0.0585}\\
 10                               &53.9864                            &6.3768             &\textbf{0.0468}\\
 11                               &43.2613                            &5.9877             &\textbf{0.0374}\\
 12                               &29.7542                            &4.5136             &\textbf{0.0302}\\\hline
\end{tabular}\vspace{-0.2in}
\end{center}
\end{table}

The results in Theorem~\ref{thm:isodp} imply that our IsoDP may possess the ability to attain a solution of strictly low rank. To confirm this, we evaluate convex optimization, LRFD and IsoDP by examining the rank of the restored trajectory matrix as well as the fitting error on the observed entries. Table~\ref{tb:motion} shows the evaluation results. It can be seen that, while the restored matrices have the same rank, the fitting error produced by IsoDP is much smaller than the competing methods. The error of convex optimization is quite large, because the method cannot produce a solution of exactly low rank unless a biased regularization parameter is chosen. Figure~\ref{fig:dinosaur:res} shows some examples of the originally incomplete and fully restored trajectories. Our IsoDP method can approximately recover the circle-like trajectories.
\subsection{Results on Movie Ratings}
We also consider the MovieLens~\cite{Harper:2015} datasets that are widely used in research and industry. The dataset we use is consist of 100,000 ratings (integers between 1 and 5) from 943 users on 1682 movies. The distribution of the observed entries is severely imbalanced: The number of movies rated by each user ranges from 20 to 737, and the number of users who have rated for each movie ranges from 1 to 583. We remove the users that have less than 80 ratings, and so for the movies. Thus the final dataset used for experiments is consist of 14,675 ratings from 231 users on 206 movies. For the sake of quantitative evaluation, we randomly select 1468 ratings as the testing data, i.e., those ratings are intentionally set unknown to the matrix completion methods. So, the percentage of the observed entries used as inputs for matrix completion is only 27.75\%.
\begin{table}
\caption{MSE on the MovieLens dataset. The regularization parameters of the competing methods have been manually tuned to the best. Here, the MSE values are evaluated on the testing data.}\label{tb:movie}\vspace{-0.2in}
\begin{center}
\begin{tabular}{cc}\hline
methods                   & MSE  \\\hline
random & 3.7623\\
average & 1.6097 \\
convex optimization & 0.9350\\
LRFD & 0.9213 \\\hline
IsoDP ($\lambda=0.0005$) & 0.8412 \\
IsoDP ($\lambda=0.0008$) & 0.8250 \\
IsoDP ($\lambda=0.001$) & \textbf{0.8228}\\
IsoDP ($\lambda=0.002$) & 0.8295\\
IsoDP ($\lambda=0.005$) & 0.8583\\\hline
\end{tabular}\vspace{-0.2in}
\end{center}
\end{table}

Despite convex optimization and LRFD, we also consider two ``trivial'' baselines: One is to estimate the unseen ratings by randomly choosing an integer from the range of 1 to 5, the other is to simply use the average rating of 3 to fill the unseen entries. The comparison results are shown in Table~\ref{tb:movie}. As we can see, all the considered matrix completion methods outperform distinctly the trivial baselines, illustrating that matrix completion is beneficial on this dataset. In particular, IsoDP with proper parameters performs much better than convex optimization and LRFD, confirming the effectiveness of IsoDP on realistic datasets.
\section{Conclusion}\label{sec:con}
This work studied the identifiability of real-valued matrices under the convex of deterministic sampling. We established two deterministic conditions, isomerism and relative well-conditionedness, for ensuring that an arbitrary matrix is identifiable from a subset of the matrix entries. We first proved that the proposed conditions can hold even if the missing data pattern is irregular. Then we proved a series of theorems for missing data recovery and convex/nonconvex matrix completion. In general, our results could help to understand the completion regimes of arbitrary missing data patterns, providing a basis for investigating the other related problems such as data forecasting.
\section*{Acknowledgement}
This work is supported in part by New Generation AI Major Project of Ministry of Science and Technology under Grant SQ2018AAA010277, in part by national Natural Science Foundation of China (NSFC) under Grant 61622305, Grant 61532009 and Grant 71490725, in part by Natural Science Foundation of Jiangsu Province of China (NSFJPC) under Grant BK20160040, in part by SenseTime Research Fund.


\begin{thebibliography}{10}
\bibitem{tao:2009:mc}
E.~Cand{\`e}s and T.~Tao, ``The power of convex relaxation: Near-optimal matrix
  completion,'' \emph{IEEE Transactions on Information Theory}, vol.~56, no.~5,
  pp. 2053--2080, 2009.

\bibitem{CandesPIEEE}
E.~Cand{\`e}s and Y.~Plan, ``Matrix completion with noise,'' in \emph{IEEE
  Proceeding}, vol.~98, 2010, pp. 925--936.

\bibitem{Mohan:2010:isit}
K.~Mohan and M.~Fazel, ``New restricted isometry results for noisy low-rank
  recovery,'' in \emph{IEEE International Symposium on Information Theory},
  2010, pp. 1573--1577.

\bibitem{rahul:jlmr:2010}
R.~Mazumder, T.~Hastie, and R.~Tibshirani, ``Spectral regularization algorithms
  for learning large incomplete matrices,'' \emph{Journal of Machine Learning
  Research}, vol.~11, pp. 2287--2322, 2010.

\bibitem{akshay:2013:nips}
A.~Krishnamurthy and A.~Singh, ``Low-rank matrix and tensor completion via
  adaptive sampling,'' in \emph{Neural Information Processing Systems}, 2013,
  pp. 836--844.

\bibitem{william:2014:nips}
W.~E. Bishop and B.~M. Yu, ``Deterministic symmetric positive semidefinite
  matrix completion,'' in \emph{Neural Information Processing Systems}, 2014,
  pp. 2762--2770.

\bibitem{raghunandan:jmlr:2010}
R.~H. Keshavan, A.~Montanari, and S.~Oh, ``Matrix completion from noisy
  entries,'' \emph{Journal of Machine Learning Research}, vol.~11, pp.
  2057--2078, 2010.

\bibitem{raghunandan:tit:2010}
------, ``Matrix completion from a few entries,'' \emph{{IEEE} Transactions on
  Information Theory}, vol.~56, no.~6, pp. 2980--2998, 2010.

\bibitem{troy:2013:nips}
T.~Lee and A.~Shraibman, ``Matrix completion from any given set of
  observations,'' in \emph{Neural Information Processing Systems}, 2013, pp.
  1781--1787.

\bibitem{Candes:2009:math}
E.~Cand{\`e}s and B.~Recht, ``Exact matrix completion via convex
  optimization,'' \emph{Foundations of Computational Mathematics}, vol.~9,
  no.~6, pp. 717--772, 2009.

\bibitem{Candes:2009:JournalACM}
E.~J. Cand\`{e}s, X.~Li, Y.~Ma, and J.~Wright, ``Robust principal component
  analysis?'' \emph{Journal of the ACM}, vol.~58, no.~3, pp. 1--37, 2011.

\bibitem{xu:2012:tit}
H.~Xu, C.~Caramanis, and S.~Sanghavi, ``Robust {PCA} via outlier pursuit,''
  \emph{{IEEE} Transactions on Information Theory}, vol.~58, no.~5, pp.
  3047--3064, 2012.

\bibitem{sun:2016:tit}
R.~Sun and Z.-Q. Luo, ``Guaranteed matrix completion via non-convex
  factorization,'' \emph{IEEE Transactions on Information Theory}, vol.~62,
  no.~11, pp. 6535 -- 6579, 2016.

\bibitem{tpami_2013_lrr}
G.~Liu, Z.~Lin, S.~Yan, J.~Sun, Y.~Yu, and Y.~Ma, ``Robust recovery of subspace
  structures by low-rank representation,'' \emph{IEEE Transactions on Pattern
  Recognition and Machine Intelligence}, vol.~35, no.~1, pp. 171--184, 2013.

\bibitem{Jain:2014:nips}
P.~Netrapalli, U.~N. Niranjan, S.~Sanghavi, A.~Anandkumar, and P.~Jain,
  ``Non-convex robust {PCA},'' in \emph{Advances in Neural Information
  Processing Systems}, 2014, pp. 1107--1115.

\bibitem{liu:tpami:2016}
G.~Liu, H.~Xu, J.~Tang, Q.~Liu, and S.~Yan, ``A deterministic analysis for
  {LRR},'' \emph{IEEE Transactions on Pattern Recognition and Machine
  Intelligence}, vol.~38, no.~3, pp. 417--430, 2016.

\bibitem{zhao:nips:2015}
T.~Zhao, Z.~Wang, and H.~Liu, ``A nonconvex optimization framework for low rank
  matrix estimation,'' in \emph{Neural Information Processing Systems}, 2015,
  pp. 559--567.

\bibitem{ge:nips:2016}
R.~Ge, J.~D. Lee, and T.~Ma, ``Matrix completion has no spurious local
  minimum,'' in \emph{Neural Information Processing Systems}, 2016, pp.
  2973--2981.

\bibitem{ruslan:2010:nips}
R.~Salakhutdinov and N.~Srebro, ``Collaborative filtering in a non-uniform
  world: Learning with the weighted trace norm,'' in \emph{Neural Information
  Processing Systems}, 2010, pp. 2056--2064.

\bibitem{Meka:2009:MCP}
R.~Meka, P.~Jain, and I.~S. Dhillon, ``Matrix completion from power-law
  distributed samples,'' in \emph{Neural Information Processing Systems}, 2009,
  pp. 1258--1266.

\bibitem{Kiraly:2012:icml}
F.~Kir\'{a}ly and R.~Tomioka, ``A combinatorial algebraic approach for the
  identi ability of low-rank matrix completion,'' in \emph{International
  Conference on Machine Learning}, 2012, pp. 2056--2064.

\bibitem{Kiraly:2015:jmlr}
F.~Kir\'{a}ly, L.~Theran, and R.~Tomioka, ``The algebraic combinatorial
  approach for low-rank matrix completion,'' \emph{J. Mach. Learn. Res.},
  vol.~16, no.~1, pp. 1391--1436, Jan. 2015.

\bibitem{Negahban:2012:JMLR}
S.~Negahban and M.~J. Wainwright, ``Restricted strong convexity and weighted
  matrix completion: Optimal bounds with noise,'' \emph{Journal of Machine
  Learning Research}, vol.~13, pp. 1665--1697, 2012.

\bibitem{JMLR:v16:chen15b}
Y.~Chen, S.~Bhojanapalli, S.~Sanghavi, and R.~Ward, ``Completing any low-rank
  matrix, provably,'' \emph{Journal of Machine Learning Research}, vol.~16, pp.
  2999--3034, 2015.

\bibitem{daniel:2016:jstsp}
D.~L. Pimentel{-}Alarc{\'{o}}n, N.~Boston, and R.~D. Nowak, ``A
  characterization of deterministic sampling patterns for low-rank matrix
  completion,'' \emph{J. Sel. Topics Signal Processing}, vol.~10, no.~4, pp.
  623--636, 2016.

\bibitem{liu:nips:2017}
G.~Liu, Q.~Liu, and X.-T. Yuan, ``A new theory for matrix completion,'' in
  \emph{Neural Information Processing Systems}, 2017, pp. 785--794.

\bibitem{Zhang06}
Y.~Zhang, ``When is missing data recoverable?'' \emph{CAAM Technical Report
  TR06-15}, 2006.

\bibitem{Donoho:spark:2003}
D.~L. Donoho and M.~Elad, ``Optimally sparse representation in general
  (nonorthogonal) dictionaries via $\ell_1$ minimization,'' \emph{Proceedings
  of the National Academy of Sciences}, vol. 100, no.~5, pp. 2197--2202, 2003.

\bibitem{liu:tsp:2016}
G.~Liu and P.~Li, ``Low-rank matrix completion in the presence of high
  coherence,'' \emph{IEEE Transactions on Signal Processing}, vol.~64, no.~21,
  pp. 5623--5633, 2016.

\bibitem{Chistov:1984}
A.~L. Chistov and D.~Grigoriev, ``Complexity of quantifier elimination in the
  theory of algebraically closed fields,'' in \emph{Proceedings of the
  Mathematical Foundations of Computer Science}, 1984, pp. 17--31.

\bibitem{Recht2008}
B.~Recht, W.~Xu, and B.~Hassibi, ``Necessary and sufficient conditions for
  success of the nuclear norm heuristic for rank minimization,'' CalTech, Tech.
  Rep., 2008.

\bibitem{Shang:2016:SAT}
F.~Shang, Y.~Liu, and J.~Cheng, ``Scalable algorithms for tractable schatten
  quasi-norm minimization,'' in \emph{{AAAI} Conference on Artificial
  Intelligence}, 2016, pp. 2016--2022.

\bibitem{xu:2017:aai}
C.~Xu, Z.~Lin, and H.~Zha, ``A unified convex surrogate for the
  schatten-\emph{p} norm,'' in \emph{{AAAI} Conference on Artificial
  Intelligence}, 2017, pp. 926--932.

\bibitem{proximal:2009:mp}
H.~Attouch and J.~Bolte, ``On the convergence of
  the proximal algorithm for nonsmooth functions involving analytic
  features,'' \emph{Mathematical Programming},
  vol. 116, no. 1-2, pp. 5--16, 2009.

\bibitem{Bolte2014}
J.~Bolte, S.~Sabach, and M.~Teboulle, ``Proximal alternating linearized
  minimization for nonconvex and nonsmooth problems,'' \emph{Mathematical
  Programming}, vol. 146, no.~1, pp. 459--494, 2014.

\bibitem{svt:cai:2008}
J.~Cai, E.~Candes, and Z.~Shen, ``A singular value thresholding algorithm for
  matrix completion,'' \emph{SIAM J. on Optimization}, vol.~20, no.~4, pp.
  1956--1982, 2010.

\bibitem{siam_2010_minirank}
B.~Recht, M.~Fazel, and P.~Parrilo, ``Guaranteed minimum-rank solutions of
  linear matrix equations via nuclear norm minimization,'' \emph{SIAM Review},
  vol.~52, no.~3, pp. 471--501, 2010.

\bibitem{siam:stewart:1969}
G.~W. Stewart, ``On the continuity of the generalized inverse,'' \emph{SIAM
  Journal on Applied Mathematics}, vol.~17, no.~1, pp. 33--45, 1969.

\bibitem{book:convex}
R.~Rockafellar, \emph{Convex {A}nalysis}.\hskip 1em plus 0.5em minus
  0.4em\relax Princeton, NJ, USA: Princeton University Press, 1970.

\bibitem{liu:tip:2014}
G.~Liu, S.~Chang, and Y.~Ma, ``Blind image deblurring using spectral properties
  of convolution operators,'' \emph{IEEE Transactions on Image Processing},
  vol.~23, no.~12, pp. 5047--5056, 2014.

\bibitem{Chen:2015:tit}
Y.~Chen, ``Incoherence-optimal matrix completion,'' \emph{IEEE Transactions on
  Information Theory}, vol.~61, no.~5, pp. 2909--2923, 2015.

\bibitem{zheng:cvpr:2012}
Y.~{Zheng}, G.~{Liu}, S.~{Sugimoto}, S.~{Yan}, and M.~{Okutomi}, ``Practical
  low-rank matrix approximation under robust l1-norm,'' in \emph{IEEE
  Conference on Computer Vision and Pattern Recognition}, 2012, pp. 1410--1417.

\bibitem{Harper:2015}
F.~M. Harper and J.~A. Konstan, ``The movielens datasets: History and
  context,'' \emph{ACM Transactions on Interactive Intelligent Systems},
  vol.~5, no.~4, pp. 1901--1919, 2015.
\end{thebibliography}
\end{document}